\documentclass[submission,copyright,creativecommons]{eptcs}
\providecommand{\event}{AiML 2026} 

\usepackage{aiml26}

\usepackage{iftex}

\ifpdf
  \usepackage{underscore}         
  \usepackage[T1]{fontenc}        
\else
  \usepackage{breakurl}           
\fi

\title{Taming Complexity in Intuitionistic Modal Logic: The Case of \fik and Its Shallow Calculus}
\author{
Han Gao
\institute{Institute of Computer Science\\
Czech Academy of Sciences}
\email{\quad gao@cs.cas.cz}
\and
Nicola Olivetti
\institute{LIS\\
Aix-Marseille University, CNRS}
\email{\quad nicola.olivetti@lis-lab.fr}
}

\newcommand{\titlerunning}{Taming Complexity in Intuitionistic Modal Logic}
\newcommand{\authorrunning}{H. Gao \& N. Olivetti}

\hypersetup{
  bookmarksnumbered,
  pdftitle    = {\titlerunning},
  pdfauthor   = {\authorrunning},
  pdfsubject  = {EPTCS},               
  pdfkeywords = {keyword1, keyword2} 
}

\usepackage{amsmath, amssymb}
\usepackage{amsfonts}

\usepackage{tikz}
\usepackage{ragged2e}
\usepackage{xhfill}
\usepackage{listings}
\usepackage{virginialake}
\usepackage{dsfont}
\usepackage{rotating}
\usepackage[linesnumbered,ruled,vlined]{algorithm2e}
\usepackage{bussproofs}
\usepackage{adjustbox}
\usepackage{virginialake}
\usepackage{thm-restate}

\usepackage{framed}

\usepackage{makecell}

\usepackage{enumitem}

\newlist{enumeratee}{enumerate}{5}
\setlist[enumeratee]{label*=\arabic*.,leftmargin=*,font=\it}

\def\+#1{\mathcal{#1}}

\newcommand{\iml}{\textsf{IML}\xspace}
\newcommand{\ipc}{\mathsf{IPC}\xspace}
\newcommand{\fik}{\textbf{FIK}\xspace}

\newcommand{\lik}{\textbf{LIK}\xspace}
\newcommand{\ik}{\textbf{IK}\xspace}
\newcommand{\ck}{\textbf{CK}\xspace}
\newcommand{\wk}{\textbf{WK}\xspace}

\newcommand{\mfik}{\mathbf{FIK}\xspace}

\newcommand{\mck}{\mathbf{CK}\xspace}

\newcommand{\basic}{\+C_0\xspace}

\newcommand{\hfik}{\+C_\mathbf{FIK}\xspace}

\newcommand{\remove}[1]{}

\newcommand\mblock[1]{[#1]}

\newcommand{\fc}{\mathsf{fc}}
\newcommand{\bc}{\mathsf{bc}}

\newcommand{\axfont}[1]{\mathsf{#1}}
\newcommand{\ax}[1]{\axfont{k}#1}
\newcommand{\modp}{\axfont{mp}}
\newcommand{\nec}{\axfont{nec}}

\newcommand{\fm}[1]{\mathsf{Fm}{(#1)}}

\newcommand{\md}[1]{\mathsf{md}{(#1)}}
\newcommand{\subfm}[1]{\mathsf{Sub}{(#1)}}

\EnableBpAbbreviations

\newcommand{\cut}{\mathsf{cut}}
\newcommand{\mcut}{\cut^{\mblock{\cdot}}}

\newcommand{\inp}{\mathsf{Input}}
\newcommand{\outp}{\mathsf{Output}}

\newcommand{\lcfik}{$\mathsf{SC}_\mathbf{FIK}$\xspace}
\newcommand{\cumlcfik}{$\mathsf{SC}_\mathbf{FIK}^\times$\xspace}

\newcommand{\mcumlcfik}{\mathsf{SC}_\mathbf{FIK}^\times\xspace}

\newcommand{\upblock}{_{\mblock{\cdot}}}

\newcommand{\wCD}{\mathsf{wCD}}

\newcommand{\At}{\mathsf{At}}
\newcommand{\FS}{\mathsf{FS}}

\def\lef#1{#1^\bullet}
\def\rig#1{#1^\circ}
\def\dow#1{{#1}^\ast}

\newcommand{\id}{\mathsf{id}}

\begin{document}

\maketitle

\begin{abstract}
Intuitionistic modal logics (\iml) comprise many systems: from constructive modal logics such as \ck and Wijesekera's \wk to Fischer Servi/Simpson's \ik, as well as some recently introduced variants. All of them are characterized by bi-relational semantics and have complete axiomatisations. However, from the perspective of proof theory and complexity, there are strong differences: while for constructive modal logics simple Gentzen calculi suffice, for \ik more complex calculi, based on nested or labeled sequents, are needed. As a consequence, the decision problem for constructive modal logics has a \textsc{Pspace} upper bound, whereas for \ik is not known and it is even conjectured to be non-elementary.
We study here the proof theory and complexity of \fik, a natural intuitionistic modal logic recently introduced. 
\fik is strictly in between \wk and \ik, yet it has the same forcing conditions as \ik. We define a ``shallow'' sequent calculus for \fik which is a nested sequent calculus where sequents have at most one level of nesting. We prove its syntactic completeness by showing the admissibility of cut. By means of this calculus we show that decision problem for \fik is in \textsc{Expspace}, whence significantly lower than the complexity conjectured for \ik.
\end{abstract}

\section{Introduction}\label{sec:intro}

Intuitionistic modal logics (\iml for short) have been studied since 
the late 40's~\cite{Fitch:1948} 
and comprise  
a variety of systems. It is usual to divide them  into two main subfamilies. In the first family, sometimes referred as \textit{intuitionistic modal logic} in \emph{strictu sensu}, the basic system is \ik, this logic was introduced by Fischer Servi \cite{Fischer:Servi:1977,Fischer:Servi:1978,Fischer:Servi:1984} and later systematized by Simpson \cite{simpson1994proof};  \ik is specifically  motivated by its  meta-theoretical properties, as it has a standard translation into first-order intuitionistic logic (FOIL).
The other family is called \textit{constructive modal logic}. It was pioneered by Prawitz \cite{Prawitz:1965} and then developed by Wijesekera who proposed the first order system \textit{constructive concurrent dynamic logic} \cite{Wijesekera:1990}, whose propositional fragment is now usually named \wk. 
The most known representative of this family is \ck as `constructive K' proposed in \cite{Bierman:dePaiva:2000}, namely a slight weaker system than \wk. 
The main motivation for constructive modal logics lies in their applications in computer science and proof theory, first of all their computational interpretation in terms of types and the Curry-Howard correspondence, but also for representing contextual reasoning~\cite{mendler2005constructive,Bellin:et:al:2001}. All the mentioned logics are nonetheless meant to be the intuitionistic counterpart of 
the classical K. 

However, this overview is not exhaustive; 
the landscape of \iml is  richer than one might expect. In recent years other systems have been introduced  with their own motivations. 
We consider in this paper one of the recent variants, \textit{forward-confluent} intuitionistic modal logic \fik, proposed in \cite{csl-2024-fik}. This logic enjoys a natural bi-relational semantics and it is the weakest logic defined by Fischer Servi (or Simpson's) forcing conditions for \emph{both} modalities. Moreover, \fik satisfies all the criteria stated by Simpson in \cite{simpson1994proof}, except for its interpretability in FOIL; for this reason it seems a meaningful system in the realm of IML on its own. Notice that the logic \fik is strictly stronger than constructive logics \ck and \wk and strictly weaker than \ik, these (strict) inclusions hold notably already for the $\Diamond$-free fragment of the respective logics. 
Other logics have been recently investigated, as we cannot be exhaustive here we just mention the logic \lik of \emph{local} modalities \cite{balbiani2024local} which is stronger than \fik  and incomparable with \ik; and the `minimal' logic in \cite{balbiani2026minimal-iml}
with a `dual' forcing condition for $\Diamond$ which that makes it incomparable with  \wk. Furthermore, a systematic investigation of other variants between \ck and \ik is carried on recently in \cite{degroot-2025semantical} where, among others, an exact correspondence between each axiom of \ik and independent frame conditions is established. 

From a semantic perspective, \emph{all} the mentioned systems are characterized uniformly by the bi-relational semantics with suitable forcing and frame conditions; in contrast, their proof theory, whenever it exists, is far from 
uniformity. 
The  constructive modal logics \ck and \wk enjoy cut-free simple Gentzen calculi, analogous to basic calculi for classical modal logic known at least since Fitting's seminal work \cite{fitting2013proof}. 
On the other hand, for \ik no such simple calculi are available and it seems unlikely that they indeed exist. To the best of our knowledge, existing sequent calculi for \ik all employ sequent with enriched structure, such as nested calculus \cite{2013cut,Galmiche:Salhi:2010,kuznets2019maehara}, fully labelled calculus \cite{Marin:et:al:2021,Girlando:2024:wollic-ik} and a more recent bi-nested calculus \cite{gao:2025-ik}. 
From a computational viewpoint the difference is even more striking: by virtue of their simple calculi, it is folklore that constructive modal logic \ck and \wk are decidable in polynomial space; In contrast, according to e.g. \cite{odintsov2008inconsistency}, the decision problem for \ik might not have an \emph{elementary} upper bound. 
It is then natural to 
ask 
where \fik stands from a complexity perspective, since it has the same forcing conditions of \ik: 
is it a neighbor of constructive modal logics or on the contrary, a neighbor of \ik? 
In \cite{csl-2024-fik} a proof system  for \fik is proposed in the form of a `bi-nested' calculus; although this calculus provides a decision procedure for \fik as well as a finite counter-model construction, it is not obvious how to obtain an upper bound of decidability by means of the calculus. In \cite{balbiani2026-fik-decidability}, alternative proofs of decidability of \fik are provided, but the complexity of the decision problem remains open. 

In this paper we tackle the problem of the complexity of \fik. By means of a new calculus, we show that the decision problem of the logic is in \textsc{Expspace}. This upper bound, although higher than that of constructive modal logics, is    significantly lower than what is conjectured for \ik, as it is elementary. In our opinion this result provides a further argument in favour of \fik as an interesting member of the \iml family. 
The question whether this bound is strict, that is of the lower bound for decidability of \fik, is open at present, as it is for \ik. 
Our result is obtained through proof-theoretic methods. We introduce a new calculus for \fik, 
which makes use of structured sequents that are only slightly more complex than standard Gentzen sequents, yet notably simpler than (bi-)nested sequents. 
A structured sequent is a Gentzen sequent possibly coupled with a tuple of `blocks' containing Gentzen sequents. As a difference with (bi)-nested calculi, structured sequents have therefore at most two-levels, no further nesting needed. The rules operate either on the `top' level or within a block. 
For this reason, we qualify this calculus as `shallow', following the terminology of \cite{gore2011correspondence}.  
Semantically, a `shallow' sequent can be viewed as a partial view  of a (bi-)nested sequent, representing only the current world and its direct accessible successors. 
We establish the completeness of the shallow calculus syntactically via a proof of cut-admissibility. 
Beyond the current results, we believe it is worthwhile to study structured calculi of this type, whose structure lies halfway between simple Gentzen calculi and full (bi-)nested calculi, for other intuitionistic modal logics with the aim of establishing complexity bounds and other properties. 

\section{Syntax and semantics}\label{sec:logics}

In this section we introduce basics of \fik, defining the syntax, the bi-relational semantics as well as its Hilbert axiom system \cite{csl-2024-fik}. 

\begin{definition}[Formulas]
	Let $\At$ be a countable set of propositional variables denoted as $p, q$, etc. The set of formulas $\+L$ whose elements are denoted $A, B$, etc., is 
	generated by the grammar 
	$$
	A::=~p~|~\bot~|~(A\wedge A)~|~(A\vee A)~|~(A\supset A)~|~\Box A~|~\Diamond A
	$$
	where $p$ ranges over $\At$. Negation $\neg A$ is defined as $A\supset\bot$.
\end{definition}

For a formula $A$, the \emph{size} of it denoted $|A|$ is the number of symbols in $A$. The \emph{modal degree} of $A$ denoted $\md{A}$ is defined as usual. 

Semantics of the logic \fik is defined in terms of bi-relational Kripke models\cite{csl-2024-fik}. 

\begin{definition}[Bi-relational frames and models]
	A \emph{bi-relational frame} is a triple $(W,\leq,R)$ where $W$ is a nonempty set of {\em worlds}, $\leq$ is a pre-order on $W$ and $R$ is a binary relation on $W$. 
	A \emph{bi-relational model} is a quadruple $(W,\leq,R,V)$, where $(W,\leq,R)$ is a bi-relational frame and $V: \At\to\+P(W)$ is a valuation function such that for all $p\in\At$, $V(p)$ is upper closed with respect to $\leq$. 

    We postulate 
    the following frame condition called forward confluence $(\fc)$ (cf. Figure \ref{fig:fc-bc} (left)): 
	\begin{center}
		$(\fc)$ \quad $\forall x, x'\!, z\in W$, if $x\leq x'$\! and $Rxz$, there is $z'\!\in W$ s.t. $Rx'z'$\! and $z\leq z'$.
	\end{center}
	A $\mfik$-frame (abbreviated as frame) is a bi-relational frame satisfying  $(\fc)$ and a $\mfik$-model (abbreviated as model) is a bi-relational model based on a $\mfik$-frame. 
\end{definition}

\begin{figure}[!t]
	\begin{center}
    \begin{tikzpicture}
    \node(x)[circle,draw,inner sep=0pt,minimum size=1mm,fill=black][label=left:$x$] at (1,2.5){};
    \node(z)[circle,draw,inner sep=0pt,minimum size=1mm,fill=black][label=right:$z$] at (3,2.5){};
    \node(y)[circle,draw,inner sep=0pt,minimum size=1mm,fill=black][label=left:$x'$] at (1,4.5){};
    \node(1)[][] at (2,3.5){$\fc$};
    \node(u)[circle,draw,inner sep=0pt,minimum size=1mm,fill=black][label=right:$z'$] at (3,4.5){};

    \draw[-stealth] (x) edge node[below] {$R$} (z);
    \draw[-stealth] (x) edge node[left] {$\leq$} (y);

    \draw[dashed, -stealth] (y) edge node[above] {$R$} (u);
    \draw[dashed, -stealth] (z) edge node[right] {$\leq$} (u);

    \node(x3)[circle,draw,inner sep=0pt,minimum size=1mm,fill=black][label=left:$x$] at (6,2.5){};
    \node(z3)[circle,draw,inner sep=0pt,minimum size=1mm,fill=black][label=right:$z$] at (8,2.5){};
    \node(u3)[circle,draw,inner sep=0pt,minimum size=1mm,fill=black][label=right:$z'$] at (8,4.5){};
    \node(3)[][] at (7,3.5){$\bc$};

    \node(y3)[circle,draw,inner sep=0pt,minimum size=1mm,fill=black][label=left:$x'$] at (6,4.5){};
    \draw[dashed, -stealth] (y3) edge node[above] {$R$} (u3);
    \draw[dashed, -stealth] (x3) edge node[left] {$\leq$} (y3);

    \draw[-stealth] (x3) edge node[below] {$R$} (z3);
    \draw[-stealth] (z3) edge node[right] {$\leq$} (u3);
    \end{tikzpicture}
	\end{center}
    \caption{Forward confluence (left) and backward confluence (right)}
	\label{fig:fc-bc}
\end{figure}
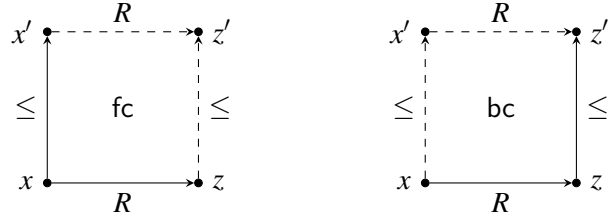

Next we define the forcing conditions for the formulas in a \fik-model. 

\begin{definition}[Forcing conditions for \fik]\label{definition:forcing:relation-prop}
    Let $\+M =(W,\leq,R,V)$ be a \fik-model. The forcing conditions of a formula at a world $w\in W$ of $\+M$ are defined as follows:
	\begin{center}
    \begin{tabular}{l l p{.7\textwidth}}
	    $\+M,w\nVdash\bot$ & & \\
        $\+M,w\Vdash p$ &iff~~& $w\in V(p)$\\
        $\+M,w\Vdash B\wedge C$ &iff~~& $\+M,w\Vdash B$ and $\+M,w\Vdash C$\\
        $\+M,w\Vdash B\vee C$ &iff~~& $\+M,w\Vdash B$ or $\+M,w\Vdash C$\\
        $\+M,w\Vdash B\supset C$ &iff~~&for all $w'\in W$ with $w\leq w'$, if $\+M,w'\Vdash B$, then $\+M,w'\Vdash C$\\
        $\+M,w\Vdash \square B$ &iff~~& for all $w',v\in W$ with $w\leq w'$ and 
        $Rw'v$, it holds $\+M,v\Vdash B$\\
        $\+M,w\Vdash \Diamond B$ &iff~~& there is $v\in W$ s.t. $Rwv$ and $\+M,v\Vdash B$\\
    \end{tabular}
	\end{center}
    We shall abbreviate $\+M,w\Vdash A$ as $w\Vdash A$ if the model is clear from the context. 
    A formula $A$ in $\+L$ is \textit{valid}, denoted $\Vdash A$, if for any bi-relational model $\mathcal{M}$ and any world $w$ in it, it holds that $\mathcal{M},w\Vdash A$.
\end{definition}
We observe that the forward confluence condition is needed to ensure the hereditary property for $\Diamond$-formulas. 
\footnote{By hereditary property, we mean the following: for every formula $A\in\+L$ if $w\Vdash A$ and $w\leq w'$ then $w'\Vdash A$. This property is required to extend the validity of intuitionistic axioms to the whole modal language $\+L$. }
To put things in a context, the logic \fik is  related to other well-known systems in the \iml family. The logic \ik is characterized by models based on frames satisfying \emph{both} forward confluence and the following condition  
\emph{backward confluence} ($\bc$) (cf. Figure \ref{fig:fc-bc} (right)):
\begin{center}
	$(\bc)$ \quad $\forall x, z, z'\in W$, if $Rxz$ and $z\leq z'$, there is $x'\!\in W$ s.t. $Rx'z'$\! and $x\leq x'$.
\end{center}
The forcing conditions of \ik are defined in the same way as in  \fik, 
which go back to Fischer Servi \cite{Fischer:Servi:1984}. On the other hand, in bi-relational models of \ck (as well as those of \wk) no additional property of frames is assumed. The forcing condition of $\Box$ is the same as in \fik, whereas for $\Diamond$: 
\begin{quote}
	$x\Vdash_{\mck} \lozenge A$ \quad iff \quad $\forall x'\in W$ with $x\leq x'$, there is $y\in W$ s.t. $Rx'y \ \& \ y\Vdash_\mck A$.
\end{quote}
This condition is equivalent to the \textit{local} one of  \fik with respect to $\mfik$-models. 

We now recall the Hilbert-style axiom system for \fik introduced in \cite{csl-2024-fik}. 

\begin{definition}[Axiom systems]
	The Hilbert-style axiom system for $\mfik$ called $\hfik$ is defined as $\hfik=\ipc\oplus\basic$ 
	where $\ipc$ denotes an axiomatization of intuitionistic propositional logic (IPL) and $\basic$ contains axioms and rules in Figure \ref{fig:axioms-basic}. 
\end{definition}

\begin{figure}[!t]
	\centering
	\begin{minipage}{.5\textwidth}
		\centering
	  	\small
		  \begin{tabular}{ccl}
			$\ax{1}$ & \quad & $\square(A\supset B)\supset(\square A\supset\square B)$ \\
			$\ax{2}$ & \quad &  $\square(A\supset B)\supset(\Diamond A\supset\Diamond B)$\\
			$\ax{3}$ & \quad & $\Diamond(A\vee B)\supset (\Diamond A\vee\Diamond B)$ \\
			$\ax{5}$ & \quad & $\neg\Diamond\bot$\\
            $\wCD$ & \quad & $\square(A\vee B)\supset((\Diamond A\supset\square B)\supset\square B)$\\
			&&\\
			\multicolumn{3}{c}{
			  \AxiomC{$A\supset B$}
			  \AxiomC{$A$}
			  \RightLabel{$\modp$}
			  \BinaryInfC{$B$}
			  \DisplayProof
				\qquad
				\AxiomC{$A$}
				\RightLabel{$\nec$}
				\UnaryInfC{$\Box A$}
				\DisplayProof
			}
			\end{tabular}
		\normalsize
		\caption{Axiom and rules of $\basic$}
		\label{fig:axioms-basic}
	\end{minipage}%
	\begin{minipage}{.5\textwidth}
		\centering
	  	\small
		\scalebox{.8}{
		\begin{tikzpicture} 
			\node(1)[][] at (3,1.5){\bf CK};
			\node(2)[][] at (3,3){\bf WK};
			\node(4)[][] at (3,4.5){\textcolor{black}{\bf FIK}};
			\node(5)[][] at (3,6){\bf IK};
			  
			\draw[->] (1) edge (2);
			\draw[->] (2) edge (4);
			\draw[->] (4) edge (5);

		\end{tikzpicture}
		}
		\normalsize
		\caption{Hierarchy of some \iml systems}
		\label{fig:Hierarchy-iml}
	\end{minipage}
\end{figure}

The notions of proof and derivability in the axiom system are defined as usual. 
The Hilbert system $\hfik$ is shown sound and complete with respect to the bi-relational semantics. 

\begin{theorem}[\cite{csl-2024-fik}]
	A formula $A\in\+L$ is $\mfik$-valid if and only if it is provable in $\hfik$. 
\end{theorem}

It is instructive to briefly recall the axiomatization of a few logics mentioned in the introduction. The logic \wk (\cite{Wijesekera:1990}) is obtained by removing $\ax{3}$ and $(\wCD)$
\footnote{The name $(\wCD)$ can be spelled as \emph{weak constant domain}, see \cite{csl-2024-fik} for a justification. } 
from $\hfik$. The logic \ck \cite{Bierman:dePaiva:2000} is obtained by further removing $\ax{5}$. 
\footnote{
There are of course other variants not included in this graph, 
e.g. a calculus for \fik excluding $(\wCD)$ is given in \cite{anupam:2023} and its semantic characterization as well as other systems like $\mck\oplus\ax{5}$ and their relations are studied in \cite{degroot-2025semantical}. } 
In turn, the axiomatization of  \ik is obtained from $\hfik$ by replacing $(\wCD)$ with the following axiom denoted  $\ax{4}$ or $(\FS)$ (for Fischer Servi):
$$(\FS/\ax{4}) \ (\Diamond A\supset \Box B)\supset \square (A \supset B)$$
Notice that $(\wCD)$ is derivable in the axiom system of  \ik \cite{csl-2024-fik}. 
The diagram in Figure \ref{fig:Hierarchy-iml} illustrates the strict inclusion relations among the mention systems. 

We end this section with some discussions on the legitimacy of \fik in the \iml family. 
Simpson \cite{simpson1994proof} proposed some criteria that a reasonable \iml should satisfy: 
(i) it is a conservative extension of IPL, (ii) it contains all instances of IPL axioms and is closed with respect to the rule $(\modp)$, (iii) $\Box$ and $\Diamond$ are not inter-definable, (iv) it satisfies disjunction property, (v) adding classical  \textit{excluded middle law}, it collapses into its classical counterpart (say K for the basic case). Lastly the (vi)$^\mathsf{th}$ criterion requires that axioms of the logic should be validated by the standard translation of modal logic into first-order intuitionistic logic (FOIL). 
As   shown in \cite{csl-2024-fik}, \fik satisfies all criteria (i)-(v), but not the sixth one. 
However in our opinion the last criterion is controversial as  the standard translation in FOIL might not be considered in itself as a main justification of a logic (alternative semantics rather than possible world relational semantics are also possible); in addition, alternative translations might be possible for \iml's which do not satisfy this criterion. Thus in this sense the absence of it should not be seen as a lack of justification. 

\section{A shallow calculus for \fik}\label{sec:shallow-calculus-fik}

In this section, we define a new calculus for \fik called a \emph{shallow} sequent calculus. The calculus makes use of nested sequents, generated by structural operators called `blocks', but as a difference with nested calculi like \cite{brunner:2009,2013cut,Galmiche:Salhi:2010,kuznets2019maehara,Fitting:2014}, there is only one level of nesting. 
This means that 
the content of a block is just an ordinary Gentzen sequent, and no further nesting is allowed. 
Consequently, most of the rules of the calculus operate on both levels: either on the `surface' or within the blocks. As in \cite{2013cut,kuznets2019maehara}, we adopt a `polarised' annotation of formulas distinguishing the `input' formulas and the unique `output' formula. Intuitively the former play the role of the antecedent of a sequent and the latter its unique consequent. 

\begin{definition}[Sequents]\label{def:shallow-sequents}
    Sequents are built out of annotated formulas in the form of $\lef A$ or $\rig A$, where a formula annotated by $\bullet$ is called an \emph{input formulas} and a formula annotated by $\circ$ is called an \emph{output formula}. 
    A \emph{simple sequent} $S$ is of the form $\Phi, \rig B$ 
    where $\Phi$ is a multi-set of $\bullet$-formulas $\lef {A_1},\ldots,\lef {A_n}$ and $B$ is a formula. 
    A \emph{shallow sequent} (or a \emph{sequent} for short) is of the form 
    $$
    \Phi,[\Psi_1],\ldots,[\Psi_l]
    $$
    where $l \geq 0$ and exactly one of $\Phi,\Psi_1,\ldots,\Psi_l$ is a simple sequent whereas the others are multi-set of $\bullet$-formulas. When $l = 0$ a sequent is a simple sequent. 
    We call $\Phi$ the top-level component and each nested component $[\Psi_i]$ a block. 
    The part of sequent 
    without the output formula is called an \emph{input} sequent. 
    We shall generically use $\Omega,\Phi,\Psi,\ldots$ to denote sequents and input sequents. 
\end{definition}

We can view shallow sequents as a simplification of nested sequents. The idea of retaining only the necessary structure from nested sequents is also present in other formalisms, such as linear nested sequents, but the simplification is orthogonal: linear nested sequents keep only a single path of worlds, whereas shallow sequents keep only the immediate successors of the current world. 

In order to define the calculus we need several operators that remove or add an output formula, transforming a sequent into an input sequent and vice versa. 

\begin{definition}
The $*$-operator is defined as follows: for a sequent $\Omega$, we denote $\dow\Omega$ as the result of removing the output formula of $\Omega$ wherever it occurs. 
The $+$-operator is defined as follows: let $\Phi$ be a simple sequent, if $\Phi$ does not contain the output formula, then $\Phi^+:=\Phi,\rig \bot$; otherwise $\Phi^+:=\Phi$. Finally we denote the formula part of $\Omega$ by $\fm{\Omega}$, obtained by removing the blocks of $\Omega$ if any. 
\end{definition}


\begin{figure}[!t]
    \begin{framed}
    \begin{center}
        \begin{adjustbox}{max width = \textwidth}
            \begin{tabular}{cc}
                \AxiomC{}
                \RightLabel{($\lef \bot_1$)}
                \UnaryInfC{$\Omega,\lef\bot$}
                \DisplayProof
                \quad
                \AxiomC{}
                \RightLabel{($\lef \bot_2$)}
                \UnaryInfC{$\Omega,\mblock{\Phi,\lef\bot}$}
                \DisplayProof
                &
                \AxiomC{}
                \RightLabel{($\id_1$)}
                \UnaryInfC{$\Omega,\lef p,\rig p$}
                \DP
                \quad
                \AxiomC{}
                \RightLabel{($\id_2$)}
                \UnaryInfC{$\Omega,\mblock{\Phi,\lef p,\rig p}$}
                \DP
                \\[.4cm]
                \AxiomC{$\Omega,\lef A,\lef B$}
                \RightLabel{($\lef \wedge_1$)}
                \UnaryInfC{$\Omega,\lef {A\wedge B}$}
                \DisplayProof
                \quad
                \AxiomC{$\Omega,\mblock{\Phi,\lef A,\lef B}$}
                \RightLabel{($\lef \wedge_2$)}
                \UnaryInfC{$\Omega,\mblock{\Phi,\lef {A\wedge B}}$}
                \DisplayProof
                &
                \AxiomC{$\Omega,\rig {A_i}$}
                \RightLabel{($\rig \vee_1$)}
                \UnaryInfC{$\Omega,\rig{A_1\vee A_2}$}
                \DisplayProof
                \quad
                \AxiomC{$\Omega,\mblock{\Phi,\rig {A_i}}$}
                \RightLabel{\rm ($\rig \vee_2$)}
                \UnaryInfC{$\Omega,\mblock{\Phi,\rig{A_1\vee A_2}}$}
                \DisplayProof
                \\[.4cm]
                \AxiomC{$\Omega,\rig A$}
                \AxiomC{$\Omega,\rig B$}
                \RightLabel{($\rig \wedge_1$)}
                \BinaryInfC{$\Omega,\rig{A\wedge B}$}
                \DisplayProof
                &
                \AxiomC{$\Omega,\mblock{\Phi,\rig A}$}
                \AxiomC{$\Omega,\mblock{\Phi,\rig B}$}
                \RightLabel{($\rig \wedge_2$)}
                \BinaryInfC{$\Omega,\mblock{\Phi,\rig{A\wedge B}}$}
                \DisplayProof
                \\[.4cm]
                \AxiomC{$\Omega,\lef A$}
                \AxiomC{$\Omega,\lef B$}
                \RightLabel{($\lef \vee_1$)}
                \BinaryInfC{$\Omega,\lef{A\vee B}$}
                \DisplayProof
                &
                \AxiomC{$\Omega,\mblock{\Phi,\lef A}$}
                \AxiomC{$\Omega,\mblock{\Phi,\lef B}$}
                \RightLabel{($\lef \vee_2$)}
                \BinaryInfC{$\Omega,\mblock{\Phi,\lef{A\vee B}}$}
                \DisplayProof
                \\[.4cm]
                \AxiomC{$\Omega^*,\lef{A\supset B},\rig A$}
                \AxiomC{$\Omega,\lef B$}
                \RightLabel{($\lef \supset_1$)}
                \BinaryInfC{$\Omega, \lef{A\supset B}$}
                \DisplayProof
                &
                \AxiomC{$\Omega^*,\mblock{\dow\Phi,\lef{A\supset B},\rig A}$}
                \AxiomC{$\Omega,\mblock{\Phi,\lef B}$}
                \RightLabel{($\lef \supset_2$)}
                \BinaryInfC{$\Omega,\mblock{\Phi, \lef{A\supset B}}$}
                \DisplayProof
                \\[.4cm]
                \AXC{$\Omega,\lef A,\rig B$}
                \RightLabel{$(\rig \supset_1)$}
                \UIC{$\Omega,\rig{A\supset B}$}
                \DP
                &
                \AXC{$\Phi,\lef A,\rig B$}
                \RightLabel{$(\rig \supset_2)$}
                \UIC{$\Omega,\mblock{\Phi,\rig{A\supset B}}$}
                \DP
                \\[.4cm]
                \AxiomC{$\Omega,\lef{\square A},\mblock{\Phi,\lef A}$}
                \RightLabel{($\lef \square$)}
                \UnaryInfC{$\Omega, \lef{\square A},\mblock{\Phi}$}
                \DisplayProof
                &
                \AXC{$\Omega,\mblock{\rig A}$}
                \RightLabel{$(\rig \square_1)$}
                \UIC{$\Omega,\rig{\square A}$}
                \DP
                \quad
                \AXC{$\Phi,\mblock{\rig A}$}
                \RightLabel{$(\rig \square_2)$}
                \UIC{$\Omega,\mblock{\Phi,\rig{\square A}}$}
                \DP
                \\[.4cm]
                \AXC{$\Omega,\mblock{\lef A}$}
                \RightLabel{$(\lef \Diamond_1)$}
                \UIC{$\Omega,\lef {\Diamond A}$}
                \DP
                \quad
                \AXC{$\Phi^+,\mblock{\lef A}$}
                \RightLabel{$(\lef \Diamond_2)$}
                \UIC{$\Omega,\mblock{\Phi,\lef {\Diamond A}}$}
                \DP
                &
                \AxiomC{$\Omega,\mblock{\Phi,\rig A}$}
                \RightLabel{($\rig \Diamond$)}
                \UnaryInfC{$\Omega,\rig{\Diamond A},\mblock{\Phi}$}
                \DP
            \end{tabular}
        \end{adjustbox}
    \end{center}
    \caption{The calculus \lcfik}
    \label{fig:shallow-fik}
    \end{framed}
\end{figure}

For example, let $\Omega=\lef A,\lef{A\supset B},\mblock{\lef{C\vee D},\rig E}$, then $\dow\Omega=\lef A,\lef{A\supset B},\mblock{\lef{C\vee D}}$. 
It is easy to see that $\dow\Omega$ is an input sequent; moreover, if $\Omega$ itself is an input sequent, then $\Omega=\dow\Omega$.  
In 
the same 
example, we have $\fm{\Omega}=\{\lef A,\lef {A\supset B}\}$.

We will need to argue about the size of sequents when discussing complexity bound. 
As we have already defined in the previous section  the size $|A|$ of a formula $A$ and its modal degree $\md{A}$, let us now extend these notions to multisets:  
the size $|\Phi|$  of a multiset of formulas $\Phi$  is defined (in the obvious way)  as the sum of the sizes of the formulas in the multi-set (counting multiplicity); 
the modal degree $\md{\Phi}:=\max\{\md{A}~|~A\in\Phi\}$. 

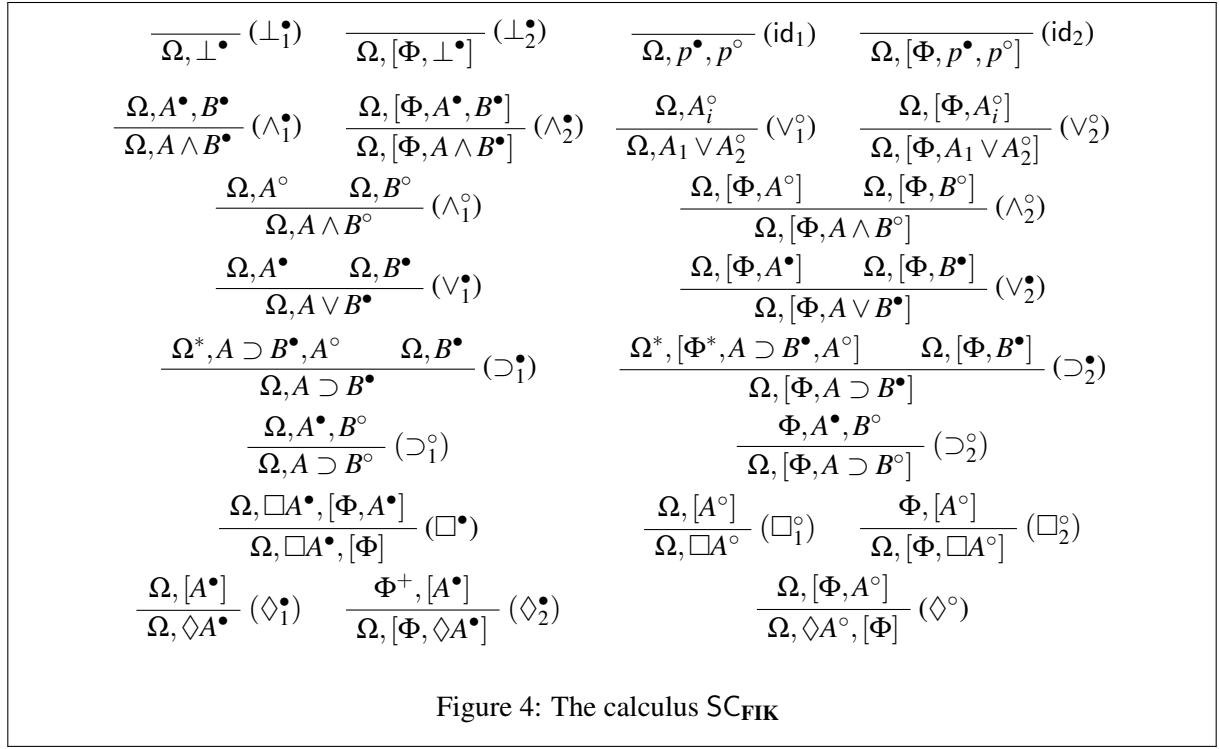
\begin{definition}[Size and modal degree of a sequent]
Given a sequent $\Omega=\Phi,[\Psi_1],\ldots,[\Psi_l]$, its size is defined as 
$|\Omega|  = |\Phi|+ |\Psi_1| + \ldots +|\Psi_l|$. 
The size of its output, denoted $|\outp(\Omega)|$, is the size of the unique output formula of $\Omega$; 
the size of its input, denoted $|\inp(\Omega)|$ is defined as $|\Omega^*|$
Moreover we define $\md{\Omega}=\max\{\md{\Phi},\md{\Psi_1}+1,\ldots,\md{\Psi_l}+1\}$. 
\end{definition}

The calculus \lcfik for \fik  operates on sequents and contains all the rules in Figure \ref{fig:shallow-fik}. Initial sequents are of four types, namely $\lef{\bot_1},\lef{\bot_2},\id_1$ and $\id_2$. 
Rules except $\lef \square$ and $\rig \Diamond$ all appear in pairs, denoted $\rig{r_1}$ and $\rig{r_2}$ (resp. $\lef{r_1}$ and $\lef{r_2}$), where the former applies to the formula part of a sequent, while the latter operates inside a block. We uniformly refer to both members of the pair as $\rig r$ (resp. $\lef r$), e.g. when we say $\rig \supset$, it means both $\rig {\supset_1}$ and $\rig {\supset_2}$. 


\begin{definition}[Formula occurrences] 
    Given an application $(r)$  of a rule with the conclusion $S$,  we say that an occurrence of a formula $A$  in $S$ is \emph{principal} if it is  explicitly treated in that application of the rule (that is reading backward the application $(r)$,   $A$ is decomposed or propagated in at least one premise).  We say that an occurrence of $A$  is  \emph{side} (or contextual) is it is not principal and it occurs in at least one premise, and it is weak if does not occur in any premises. 
\end{definition}

For example, in $(\rig {\vee_2})$, the explicit $\rig {A\vee B}$ is principal and formulas in $\Omega$ and $\Phi$ are all side; in $(\lef {\Diamond_2})$, the explicit $\lef {\Diamond A}$ is principal, formulas in $\Phi$ are side and formulas in $\Omega$ are weak. 

The notions of derivations, proofs (e.g. a derivation whose leaves are axioms namely initial sequents) and height of a derivation are defined as usual in sequent calculi. We write $\vdash_n\Omega$ to denote the proposition that $\Omega$ has a proof of height $n$ in \lcfik. Moreover, we say a formula $A$ is provable in \lcfik if so is the sequent $\rig A$. 




\begin{example}
    Axiom $(\textbf{wCD}): \square(p\vee q)\supset ((\Diamond p\supset \square q)\supset\square q)$ 
    and the formula $(\neg\square \bot\supset\square\bot)\supset\square\bot$ are provable in \lcfik. \footnote{Note that this $\Diamond$-free formula is not valid in \textbf{CK}, which makes the $\Diamond$-free fragments of these two logics different \cite{anupam:2023}. } 
    The proofs are found in Figure \ref{fig:proof-wcd}.
\end{example}

\begin{figure}[t]
    \begin{framed}
    \begin{center}
        \begin{adjustbox}{max width=\textwidth}
        \begin{tabular}{c}
        \AXC{$\lef{\square(p\vee q)}, \lef{\square q}, [\lef q, \lef{p}, \rig q]$}
        \RightLabel{$\lef {\square}$}
        \UIC{$\lef{\square(p\vee q)}, \lef{\square q}, [\lef{p}, \rig q]$}
        \AXC{$\lef{\square(p\vee q)}, \lef{\Diamond p\supset \square q}, [\lef{p},\rig p]$}
        \RightLabel{$\rig \Diamond$}
        \UIC{$\lef{\square(p\vee q)}, \lef{\Diamond p\supset \square q}, [\lef{p}],\rig{\Diamond p}$}
        \RightLabel{$\lef {\supset_1}$}
        \BIC{$\lef{\square(p\vee q)}, \lef{\Diamond p\supset \square q}, [\lef{p}, \rig q]$}
        \AXC{$\lef{\square(p\vee q)}, \lef{\Diamond p\supset \square q}, [\lef{q}, \rig q]$}
        \RightLabel{$\lef {\vee_2}$}
        \BIC{$\lef{\square(p\vee q)}, \lef{\Diamond p\supset \square q}, [\lef{p\vee q}, \rig q]$}
        \RightLabel{$\lef {\square}$}
        \UIC{$\lef{\square(p\vee q)}, \lef{\Diamond p\supset \square q}, [\rig q]$}
        \RightLabel{$\rig {\square_1}$}
        \UIC{$\lef{\square(p\vee q)}, \lef{\Diamond p\supset \square q}, \rig{\square q}$} 
        \RightLabel{$\rig{\supset_1}\times 2$}
        \UIC{$\rig{\square(p\vee q)\supset ((\Diamond p\supset \square q)\supset\square q)}$}
        \DP
        \\ \\
        \AXC{$\lef {\neg\square \bot\supset\square\bot}, \mblock{\lef \bot},\lef{\square\bot},\rig {\bot}$}
        \RightLabel{$\lef {\square}$}
        \UIC{$\lef {\neg\square \bot\supset\square\bot}, \mblock{\varnothing},\lef{\square\bot},\rig {\bot}$}
        \RightLabel{$\rig{\supset_1}$}
        \UIC{$\lef {\neg\square \bot\supset\square\bot}, \mblock{\varnothing},\rig {\neg\square\bot}$}
        \AXC{$\lef {\square\bot}, \mblock{\lef \bot, \rig\bot}$}
        \RightLabel{$\lef {\square}$}
        \UIC{$\lef {\square\bot}, \mblock{\rig\bot}$}
        \RightLabel{$\lef {\supset_1}$}
        \BIC{$\lef {\neg\square \bot\supset\square\bot}, \mblock{\rig\bot}$}
        \RightLabel{$\rig {\square_1}$}
        \UIC{$\lef {\neg\square \bot\supset\square\bot}, \rig {\square\bot}$}
        \RightLabel{$\rig{\supset_1}$}
        \UIC{$\rig {(\neg\square \bot\supset\square\bot)\supset\square\bot}$}
        \DP
        \end{tabular}
        \end{adjustbox}
    \end{center}
    \caption{Examples of proofs in \lcfik}
    \label{fig:proof-wcd}
    \end{framed}
\end{figure}

Sequents are interpreted according to the location of the output formula and only when the output formula occurs on the top-level  the sequents have a formula interpretation. 

\begin{definition}[Semantic interpretation]
    Let $\+M=(W,\leq,R,V)$ be a bi-relational model and $x,y\in W$. For a sequent $\Omega$, we define $\+M, x\Vdash \Omega$ according to the location of its output formula as follows
    \begin{enumerate}
        \item[a.] if $\Omega=\Phi,\overline{[\Psi_i]}_{i\in I},\rig F$ where $\Phi$ is a multi-set of $\bullet$-formulas, 
        \begin{center}
            $x\Vdash \Omega$~~iff~~$x\Vdash \bigwedge\Phi \wedge \bigwedge \{\Diamond \bigwedge \Psi_i\}_{i\in I}\supset F$ 
        \end{center}
        \item[b.] if $\Omega=\Phi,\overline{[\Psi_i]}_{i\in I},[\Xi,\rig F]$ where $\Phi$ is a multi-set of $\bullet$-formulas,  
        \begin{center}
            \begin{adjustbox}{max width=.9\textwidth}
            $x\Vdash \Omega$~~iff~~$\forall x'\geq x$ if $x'\Vdash \bigwedge\Phi\wedge\bigwedge \{\Diamond \bigwedge \Psi_i\}_{i\in I}$ then  $\forall y (Rx'y~\&~y\Vdash \bigwedge \Xi \
            \mbox{implies} \
            y\Vdash F)$
        \end{adjustbox}
        \end{center}
    \end{enumerate}
    $\Omega$ is called valid if $\+M,x\Vdash \Omega$ holds for any $\+M$ and any $x$. 
\end{definition}

Based on the semantic interpretation, we show the soundness of \lcfik. 

\begin{restatable}{theorem}{thmsoundness}\label{thm:soundness}
        For any sequent $\Omega$, if it is derivable in \lcfik, then it is valid in $\mfik$. 
\end{restatable}

\begin{proof}
    We prove this by induction on the structure of a derivation. It suffices to verify that, the initial sequents are valid and the inference rules preserve validity. 
    We only show some non-trivial cases.  
    \begin{description}
        \item[$(\lef \supset_2)$]  We only show the case when $\rig F\in\Phi$ as the other cases are similar. Assume w.l.o.g. that $\Omega$ is block-free and the output formula of the conclusion is in $\Omega$. Then the rule is with conclusion $\Omega^*,\mblock{\Phi, \lef{A\supset B}},\rig F$ and two premises $\Omega^*,\mblock{\dow\Phi,\lef{A\supset B},\rig A}$ and $\Omega^*,\mblock{\Phi, \lef{A\supset B}},\rig F$. 
        Assume by absurdity that this rule application does not preserve validity, then there exists a bi-relational model $\+M=(W,\leq,R,V)$ with $x\in W$ such that $x\nVdash\Omega^*,\mblock{\Phi, \lef{A\supset B}},\rig F$. By definition, there is $x'\in W$ such that $x\leq x'$ and $x'\Vdash\bigwedge\Omega\wedge \Diamond (\bigwedge \Phi\wedge (A\supset B))$ and $x'\nVdash F$. It follows that there is $y\in W$ such that $y\Vdash (A\supset B)\wedge \bigwedge\Phi$. Since the left premise is valid, we have $x'\Vdash \Omega^*,\mblock{\dow\Phi,\lef{A\supset B},\rig A}$. Thus $y\Vdash A$. 
        Recall $y\Vdash A\supset B$, so $y\Vdash B$. 
        Meanwhile, the right premise is also valid, so $x'\Vdash \Omega^*,\mblock{\Phi,\lef B},\rig F$, which implies $x'\nVdash \Diamond (\bigwedge\Phi\wedge B)$. Since $Rx'y$ and $y\Vdash \bigwedge \Phi$, we have $y\nVdash B$, a contradiction. 
        \item[$(\lef\Diamond_2)$] 
        The rule is of one of the two forms: conclusion $\Omega,[\Phi,\lef {\Diamond A},\rig F]$ and premise $\Phi,[\lef A],\rig F$; conclusion $\Phi,[\lef A],\rig \bot$ and premise $\Omega,[\Phi,\lef {\Diamond A}]$. 
        For the first case, 
        assume this rule does not preserve validity, then there exists a bi-relational model $\+M=(W,\leq,R,V)$ with $x\in W$ such that $x\nVdash \Omega,[\Phi,\lef {\Diamond A},\rig F]$. By definition, this means there is $x',y\in W$ such that $x'\geq x, Rx'y$ and $y\Vdash \bigwedge \Phi \wedge \Diamond A$ and $y\nVdash F$. Meanwhile, since the premise is valid, we have $y\Vdash \Phi,[\lef A],\rig F$, which means $y\Vdash \bigwedge \Phi\wedge  {\Diamond A}\supset F$. 
        Since $y\leq y$, by $y\Vdash \bigwedge \Phi \wedge \Diamond A$ it follows $y\Vdash F$, 
        a contradiction. 
        For the second case, the output formula occurs in $\Omega$. 
        Assume this rule application does not preserve validity, then there exists a bi-relational model $\+M=(W,\leq,R,V)$ with $x\in W$ such that $x\nVdash \Omega,[\Phi,\lef {\Diamond A}]$. By definition, regardless of the exact location of the output formula, there exists $x'\in W$ such that $x'\geq x$ such that $x'\Vdash \Diamond (\bigwedge \Phi \wedge \Diamond A)$. This further implies there is $y\in W$ such that $Rx'y$ and $y\Vdash \bigwedge \Phi \wedge \Diamond A$. 
        Meanwhile, since the premise is valid, we have $y\Vdash \Phi,[\lef A],\rig \bot$, which means $y\nVdash \bigwedge \Phi\wedge  {\Diamond A}$, 
        a contradiction. 
    \end{description}
    This completes the proof. 
\end{proof}

\section{Admissibility of cut and completeness}\label{sec:completeness}

In this section we prove \lcfik is syntactically complete, meaning that all theorems of $\hfik$ can be derived by  \lcfik. As usual we first prove that the calculus \lcfik augmented with suitable cut rules is complete. Then we show that the cut rules are admissible in \lcfik. 
In accordance with the form of the sequents, we need to consider two cut rules that operate respectively on the top level and within a block. 

\begin{definition}[Cut-rules]
    The cut rules are the following:
    $$
    \begin{aligned}
        \AxiomC{$\dow\Omega,\rig A$}
        \AxiomC{$\lef A, \Omega$}
        \RightLabel{$(\cut)$}
        \BinaryInfC{$\Omega$}
        \DP
        &
        \quad
        &
        \AxiomC{$\dow\Omega,\mblock{\dow\Phi,\rig A}$}
        \AxiomC{$\Omega,\mblock{\lef A, \Phi}$}
        \RightLabel{$(\mcut)$}
        \BinaryInfC{$\Omega,\mblock{\Phi}$}
        \DP
    \end{aligned}
    $$
    In both rules, the formula $A$ which occurs explicitly in both premisses in both polarized forms $\lef A,\rig A$ is called the \textbf{cut formula}. 
\end{definition}

We can easily prove the syntactic completeness of \lcfik augmented by the two cut rules.  

\begin{theorem}[Completeness of \lcfik +$(\cut)$+$(\mcut)$]
    If a formula $A$ is provable in $\hfik$, then it is provable in \lcfik+$(\cut)$+$(\mcut)$. 
\end{theorem}

\begin{proof}
We show that any instance of axioms of  $\hfik$ can be proved in \lcfik and   that the rule of (mp)  is simulated by $(\cut)$ as usual. 
The proof of axiom ($\wCD$) is given in Figure \ref{fig:proof-wcd} as an example. 
For rule ({\sf nec}), we prove by induction on a derivation that for any formula $A$ if $\rig A$ is provable in \lcfik +$(\cut)+(\mcut)$, $[\rig  A]$ is also provable in the same system, and then we conclude by an application of $(\rig \square_1)$. 
\end{proof}

The rest of the section is devoted to a syntactic proof of the admissibility of both cut rules. We introduce some preliminary definitions and prove some preliminary facts. 

\begin{definition}[Admissibility and invertibility]
    Let ${\cal R}$ be an inference rule  of the form $\frac{S_1 \ (S_2)}{S}$. \footnote{By $\frac{S_1 \ (S_2)}{S}$ we mean two possible rules either with one premise $S_1$ or two premises $S_1$ and $S_2$.}
    We say that ${\cal R}$ is \emph{admissible} in \lcfik if for every application of the rule $(r) = \frac{S_1 \ (S_2)}{S}$: if  $\vdash_{n_i} S_i$ for some $n_i$,  then  $\vdash_n S$ for some $n$. 
    Rule ${\cal R}$ is further called \emph{height-preserving admissible} if for every application it holds $n\leq n_i$ for $i\in\{1,2\}$. 
    Conversely, ${\cal R}$ is called \emph{invertible} in \lcfik if for every application $\frac{S_1 \ (S_2)}{S}$ whenever $\vdash_n S$ it holds $\vdash_{n_i} S_i$ for $i\in\{1,2\}$. 
    and further \emph{height-preserving invertible} if $n_i\leq n$ for $i\in\{1,2\}$. 
\end{definition}

As usual, weakening is hp-admissible by induction on the height of a derivation: 
\begin{proposition}
    The following structural rules are height-preserving admissible in \lcfik,
    $$
    \begin{aligned}
        \AxiomC{$\Omega$}
        \RightLabel{$(\lef {w_1})$}
        \UnaryInfC{$\lef A,\Omega$}
        \DP
        &
        \quad
        &
        \AxiomC{$\Omega,\mblock{\Phi}$}
        \RightLabel{$(\lef{w_2})$}
        \UnaryInfC{$\Omega,\mblock{\lef A,\Phi}$}
        \DP
        &
        \quad
        &
        \AXC{$\Omega$}
        \RightLabel{$(w\upblock^+)$}
        \UIC{$\Omega,\mblock{\Phi}$}
        \DP
    \end{aligned}
    $$
\end{proposition}

\begin{restatable}[Inversion lemma]{lemm}{lemmainver}\label{inver-lemma}
    $(\lef\wedge) (\rig\wedge)(\lef \vee)(\lef \square)(\rig{\supset_1})(\rig{\square_1})$ 
    and $(\lef{\Diamond_1})$ 
    are height-preserving invertible in \lcfik. Moreover, $(\lef \supset)$ is height-preserving invertible only for the left premise, which means if $\vdash_n \Omega,\lef{A\supset B}$ then $\vdash_n \Omega,\lef B$ and if $\vdash_n \Omega,[\Phi,\lef{A\supset B}]$ then $\vdash_n \Omega, [\Phi,\lef B]$.
\end{restatable}


According to the lemma above, the rules in \lcfik are then divided into the invertible and non-invertible groups, see Table \ref{table:invertible-rules}. Moreover, with the lemma and definition above, we show the following, 

\begin{table}[t]
\centering
\begin{tabular}{|c|c|}
    \hline
    invertible rules & $(\lef\wedge) (\rig\wedge)(\lef \vee)(\lef \square)(\rig{\supset_1})(\rig{\square_1})(\lef{\Diamond_1})$ \\[1ex]
    \hline
    non-invertible rules & $(\rig{\vee})(\rig{\supset_2})(\rig{\square_2})(\rig{\Diamond})(\lef{\Diamond_2})$\\[1ex]
    \hline
    rules with one invertible premise & $(\lef \supset)$\\[1ex]
    \hline
\end{tabular}

~

\caption{Invertible and non-invertible rules in \lcfik}
\label{table:invertible-rules}
\end{table}

\begin{restatable}{propp}{propwc}\label{prop:admissibility-weakening-contraction}
    The following structural rules are height-preserving admissible in \lcfik:
    $$
    \begin{aligned}
        \AxiomC{$\lef A,\lef A,\Omega$}
        \RightLabel{$(\lef {c_1})$}
        \UnaryInfC{$\lef A,\Omega$}
        \DP
        &
        \quad
        &
        \AxiomC{$\Omega,\mblock{\lef A,\lef A,\Phi}$}
        \RightLabel{$(\lef{c_2})$}
        \UnaryInfC{$\Omega,\mblock{\lef A,\Phi}$}
        \DP
        &
        \quad
        &
        \AXC{$\Omega,\mblock{\Phi},\mblock{\dow\Phi}$}
        \RightLabel{$(c\upblock^-)$}
        \UIC{$\Omega,\mblock{\Phi}$}
        \DP
        &
        \quad
        &
        \AXC{$\Phi$}
        \RightLabel{$(\mblock{\cdot})$}
        \UIC{$\Omega,\mblock{\Phi}$}
        \DP
    \end{aligned}
    $$
\end{restatable}

\begin{corollary}
    If $\Omega,\Omega^*$ is provable, then $\Omega$ is provable. 
\end{corollary}

Next, we introduce the following measure in order to show cut admissibility. 

\begin{definition}[rank of $\cut$-application]
   To  each application $(r)$ of $(\cut)$ (resp. $(\mcut)$) in ${\cal D}$ whose premises are provable, 
    we associate 
    $\mathbf{r}_\cut(d,h)$ (resp. $\mathbf{r}_{\mcut}(d,h)$) as the \textbf{rank} of it, where $d$ denotes 
    size 
    of the cut formula and $h$ denotes sum of the heights of the proofs of two premises of $(r)$. \footnote{Usually $d$ denotes formula complexity but it also works here by using size, since the size of a formula strictly decreases in accordance with proper subformulas.}
    We take the lexicographic order on the pair $(d,h)$ as the order on the rank $\mathbf{r}_\cut(d,h)$ (resp. $\mathbf{r}_{\mcut}(d,h)$).  
\end{definition}

\begin{restatable}[Cut-admissibility]{theorem}{thmcut}
    Both $(\cut)$ and $(\mcut)$ are admissible in \lcfik, meaning that if a sequent $\Omega$ is derivable in \lcfik$+ \ \cut+\mcut$ then it is also derivable in \lcfik.
\end{restatable}

\begin{proof}
    We 
    carry on 
    the proof by mutual induction on $\cut(d_1,h_1)$ and $\mcut(d_2,h_2)$, where 
    \begin{itemize}
        \item $\cut(d_1,h_1)$:=  every application $(r)$ of $\cut$ in a derivation 
        has a rank strictly smaller than $(d_1,h_1)$
        \item $\mcut(d_2,h_2)$:= every application $(r)$ of $\mcut$ in a derivation 
        has a rank strictly smaller than $(d_2,h_2)$
    \end{itemize}
    Therefore, the induction hypotheses we use are the following:\\

    \begin{adjustbox}{max width=\textwidth}
    \setlength{\tabcolsep}{.5mm}{
    \begin{tabular}{lll}
        \textbf{IH1} & & $\forall(d_1,h_1). 
        \Bigl[(\forall (d_2,h_2)<(d_1,h_1). \mcut (d_2,h_2))~\&~(\forall (d_1',h_1')<(d_1,h_1). \cut (d_1',h_1'))\Bigr] \to \cut(d_1,h_1)$\\[.3cm]
        \textbf{IH2} & & $\forall(d_2,h_2). \Bigl[(\forall (d_1,h_1)<(d_2,h_2). \cut (d_2,h_2))~\&~(\forall (d_2',h_2')<(d_2,h_2). \mcut (d_2',h_2'))\Bigr] \to \mcut(d_2,h_2)$\\
        \\
    \end{tabular}
        }
    \end{adjustbox}
    
    The two inductive hypotheses will take care of the two derivations of the following form, denoted as $\+D$ and $\+E$ respectively, 
    \begin{center}
        \vspace{-.5cm}
        \begin{adjustbox}{max width=\textwidth}
        \begin{tabular}{ccc}
            $
            \vlderivation{
            \vliin{}{(\cut)}{\Omega}{
            \vltr {\+D_1} {
                    \dow\Omega,\rig A
                }{
                \vlhy {\quad}}
                {
                \vlhy {}}
                {
                \vlhy {\quad}}
            }
            {
                \vltr {\+D_2} {
                    \lef A,\Omega
                }{
                \vlhy {\quad}}
                {
                \vlhy {}}
                {
                \vlhy {\quad}} 
            }
            }
            $
            &
            \quad \quad
            &
            $
        \vlderivation{
        \vliin{}{(\mcut)}{\Omega,\mblock{\Phi}}{
        \vltr {\+E_1} {
            \dow\Omega,\mblock{\dow\Phi,\rig A}
            }{
            \vlhy {\quad}}
            {
            \vlhy {}}
            {
            \vlhy {\quad}}
        }
        {
            \vltr {\+E_2} {
                \Omega,\mblock{\lef A,\Phi}
            }{
            \vlhy {\quad}}
            {
            \vlhy {}}
            {
            \vlhy {\quad}} 
        }
        }
        $
        \end{tabular}
        \end{adjustbox}
    \end{center}

    The structure of the proof is as follows. We 
    consider different cases 
    according to the role of the cut formula in the premises of $(\cut)$ for the last step in $\+D_1$ and $\+D_2$, then similarly for the roles of the cut formula in the premises of $(\mcut)$ for the last step in $\+E_1$ and $\+E_2$. Our goal is to push up each application of $(\cut)$ and $(\mcut)$ to some upper applications with strictly smaller ranks. 

    First, for $\+D$, let us denote the last rule applied in $\+D_1$ and $\+D_2$ as $(r_1),(r_2)$ respectively. If either $\rig A$ is weak in $(r_1)$ or $\lef A$ is weak in $(r_2)$, we only need to apply the same rule with another context without that occurrence of $A$. Otherwise, neither $\rig A$ nor $\lef A$ is weak, then one of the following holds: 
    (D1) one of the premises is an axiom 
    namely the inductive base ($\cut$ applied at height 0); and for inductive step, 
    (D2) $A$ is principal in \textbf{both} premises; 
    (D3) $\rig A$ is principal in $(r_1)$ and $\lef A$ is side in $(r_2)$; 
    (D4) $\rig A$ is side in $(r_1)$. 

    Next, for $\+E$, similar to the first half of the proof, we discuss possible different roles of the cut formula in the premises of $(\mcut)$ for the last step in $\+E_1$ and $\+E_2$. We denote the last rule applied in $\+E_1$ and $\+E_2$ as $(e_1),(e_2)$ respectively. In addition to the weak cases, we have: 
    (E1) one of the premises is an axiom; 
    (E2) cut formula $A$ is principal in \textbf{both} premises; 
    (E3) $\rig A$ is principal in $(e_1)$ and $\lef A$ is side in $(e_2)$; 
    (E4) $\rig A$ is side in $(e_1)$. 

    Let us consider all the cases of (D1)-(D4) and (E1)-(E4) in turn. 

    \noindent \textbf{(D1)} one of the premises is an axiom, which implies one of the following:
        (1) $\dow\Omega$ or $\Omega$ is an axiom, so is the sequent on the conclusion;
        (2) $A=p$ and $\lef p\in \Omega$, then we obtain $\Omega$ just by contraction from the right premise $\lef A,\Omega$;
        (3) $A=p$ and $\rig p\in\Omega$, dual to the second case. 

        \noindent \textbf{(D2)} cut formula $A$ is principal in both premises, we consider the following non-trivial sub-cases:
    
    \begin{itemize}
    \item If $A= B\supset C$, 
    we have 
    \begin{center}
        \vspace{-.3cm}
        \begin{adjustbox}{max width=.92\textwidth}
        \begin{tabular}{ccc}
        $
        \vlderivation{
        \vliin{}{(\cut)}{\Omega}{
        \vlin{}{(\rig \supset)}{
            \dow\Omega,\rig{B\supset C}
            }{
                \vltr {\+D_1'} {
                    \dow\Omega,\lef B,\rig C
                    }{
                    \vlhy {\quad}}
                    {
                    \vlhy {}}
                    {
                    \vlhy {\quad}}
            }
        }
        {
            \vliin{}{(\lef \supset)}{
            \lef{B\supset C},\Omega
            }{
                \vltr {\+D_2'} {
                    \lef{B\supset C}, \dow\Omega,\rig B
                    }{
                    \vlhy {\quad}}
                    {
                    \vlhy {}}
                    {
                    \vlhy {\quad}}
            }
            {
                \vltr {\+D_2''} {
                    \lef C,\Omega
                    }{
                    \vlhy {\quad}}
                    {
                    \vlhy {}}
                    {
                    \vlhy {\quad}}
            }
        }
        }
        $
        & $\leadsto$ & 
        \AXC{$\dow\Omega,\rig{B\supset C}$}
        \AXC{$\lef{B\supset C},\dow\Omega,\rig B$}
        \RightLabel{$(\cut)$}
        \BIC{$\dow\Omega,\rig B$}
        \AXC{$\dow\Omega,\lef B,\rig C$}
        \RightLabel{$(\cut)$}
        \BIC{$\dow\Omega,\rig C$}
        \AXC{$\lef{C},\Omega$}
        \RightLabel{$(\cut)$}
        \BIC{$\Omega$}
        \DP
        \end{tabular}
    \end{adjustbox}
    \end{center}
    \item If $A=\square B$, 
    further assume $\Omega=\Omega',[\Phi]$, then we have 
    \begin{center}
        \vspace{-.5cm}
        \begin{adjustbox}{max width=.92\textwidth}
        \begin{tabular}{ccc}
        $
        \vlderivation{
        \vliin{}{(\cut)}{\Omega',\mblock{\Phi}}{
        \vlin{}{(\rig \square)}{
            \dow\Omega,\rig{\square B}
            }{
                \vltr {\+D_1'} {
                    \Omega^*,\mblock{\rig B}
                    }{
                    \vlhy {\quad}}
                    {
                    \vlhy {}}
                    {
                    \vlhy {\quad}}
            }
        }
        {
            \vlin{}{(\lef \square)}{
            \Omega',\lef{\square B},\mblock{\Phi}
            }{
                \vltr {\+D_2'} {
                    \Omega',\lef{\square B},\mblock{\lef B,\Phi}
                    }{
                    \vlhy {\quad}}
                    {
                    \vlhy {}}
                    {
                    \vlhy {\quad}}
            }
        }
        }
        $
        & $\leadsto$ & 
        \AXC{$\Omega^*,\mblock{\rig B}$}
        \AXC{$\dow{\Omega'},\mblock{\dow\Phi},\rig{\square B}$}
        \RightLabel{$(\lef w\upblock)$}
        \UIC{$\dow{\Omega'},\mblock{\lef B, \dow\Phi},\rig{\square B}$}
        \AXC{$\Omega',\lef{\square B},\mblock{\lef B,\Phi}$}
        \RightLabel{$(\cut)$}
        \BIC{$\Omega',\mblock{\lef B,\Phi}$}
        \RightLabel{$(\mcut)$}
        \BIC{$\Omega,\mblock{\Phi}$}
        \DP
        \end{tabular}
    \end{adjustbox}
    \end{center}
    \item If $A=\Diamond B$, assume $\Omega=\Omega',[\Phi]$ and the output formula of the conclusion is in $\Phi$, then we have
    \begin{center}
        \vspace{-.5cm}
        \begin{adjustbox}{max width=.92\textwidth}
        \begin{tabular}{ccc}
        $
        \vlderivation{
        \vliin{}{(\cut)}{
            \Omega',\mblock{\Phi}
        }{
        \vlin{}{(\rig \Diamond)}{
            \dow{\Omega'},\mblock{\dow\Phi},\rig {\Diamond B}
            }{
                \vltr {\+D_1'} {
                    \dow{\Omega'},\mblock{\dow\Phi,\rig B}
                    }{
                    \vlhy {\quad}}
                    {
                    \vlhy {}}
                    {
                    \vlhy {\quad}}
            }
        }
        {
            \vlin{}{(\lef \Diamond)}{
            \lef {\Diamond B},\Omega',\mblock{\Phi}
            }{
                \vltr {\+D_2'} {
                    \Omega',\mblock{\Phi},\mblock{\lef B}
                    }{
                    \vlhy {\quad}}
                    {
                    \vlhy {}}
                    {
                    \vlhy {\quad}}
            }
        }
        }
        $
            & $\leadsto$ & 
            \AXC{$\dow{\Omega'},\mblock{\dow\Phi,\rig B}$}
            \RightLabel{$(w^+\upblock)$}
            \UIC{$\dow{\Omega'},\mblock{\dow\Phi,\rig B},\mblock{\dow\Phi}$}
            \AXC{$\Omega',\mblock{\Phi},\mblock{\lef B}$}
            \RightLabel{$(\lef w)$}
            \UIC{$\Omega',\mblock{\Phi},\mblock{\dow\Phi,\lef B}$}
            \RightLabel{$(\mcut)$}
            \BIC{$\Omega',\mblock{\Phi^*},\mblock{\Phi}$}
            \RightLabel{$w\upblock^-$}
            \UIC{$\Omega',\mblock{\Phi}$}
            \DP
        \end{tabular}
        \end{adjustbox}
    \end{center}
    The other case when the output formula is in $\Omega'$ is similar. 
    \end{itemize}

    \noindent \textbf{(D3)} $\rig A$ is principal in $(r_1)$ while $\lef A$ is side in $(r_2)$. We have the following sub-cases:  
        \begin{enumerate}[label=(\roman*)]
            \item $r_2$ is one of the invertible rules except $\rig{\square_1}$. 
            As an example, we show the case when $r_2=\lef{\supset_1}$. 
            Suppose $\Omega=\lef{B\supset C}, \Omega'$, in this case, we have the following derivation unfolding $\+D$, 

            \begin{center}
                \vspace{-.3cm}
                $
                \vlderivation{
                \vliin{}{(\cut)}{
                    \lef{B\supset C}, \dow{\Omega'}
                }{
                    \vltr {\+D_1'} {
                        \lef{B\supset C}, \dow{\Omega'},\rig A      
                    }{
                    \vlhy {\quad}}
                    {
                    \vlhy {}}
                    {
                    \vlhy {\quad}} 
                }
                {
                    \vliin{}{(\lef{\supset_1})}{
                        \lef A,\lef{B\supset C}, \Omega'
                    }{
                    \vltr {\+D_{21}'}{
                        \lef A,\lef{B\supset C},\dow{\Omega'},\rig B
                    }{
                    \vlhy {\quad}}
                    {
                    \vlhy {}}
                    {
                    \vlhy {\quad}} 
                    }
                    {
                        \vltr {\+D_{22}'}{
                            \lef A,\lef C, {\Omega'}
                        }{
                        \vlhy {\quad}}
                        {
                        \vlhy {}}
                        {
                        \vlhy {\quad}} 
                    }
                }
                }
                $
            \end{center}
            Then we transform the derivation into the following
            \vspace{-.2cm}
            \begin{center}
            \begin{adjustbox}{max width=.93\textwidth}
                \AXC{$\lef{B\supset C}, \dow{\Omega'},\rig A$}
                \RightLabel{$(\lef w)$}
                \UIC{$\lef{B\supset C}, \lef{B\supset C}, \dow{\Omega'},\rig A$}
                \AXC{$\lef A,\lef{B\supset C},\dow{\Omega'},\rig B$}
                \RightLabel{$(\lef w)$}
                \UIC{$\lef A,\lef{B\supset C},\lef{B\supset C},\dow{\Omega'},\rig B$}
                \RightLabel{$(\cut)$}
                \BIC{$\lef{B\supset C},\dow{\Omega'},\rig B$}
                \AXC{$\lef{B\supset C}, \dow{\Omega'},\rig A$}
                \RightLabel{$(\lef w)$}
                \UIC{$\lef C, \lef{B\supset C}, \dow{\Omega'},\rig A$}
                \AXC{$\lef A,\lef C, {\Omega'}$}
                \RightLabel{$(\lef w)$}
                \UIC{$\lef A,\lef C, \lef{B\supset C}, {\Omega'}$}
                \RightLabel{$(\cut)$}
                \BIC{$\lef C, \lef{B\supset C}, {\Omega'}$}
                \RightLabel{$(\lef{\supset_1})$}
                \BIC{$\lef {B\supset C}, \lef{B\supset C}, {\Omega'}$}
                \RightLabel{$(\lef w)$}
                \UIC{$\lef{B\supset C}, {\Omega'}$}
                \DP
            \end{adjustbox}
            \end{center}
            \item suppose $r_2=\rig{\square_1}$ and $\Omega=\Omega',\rig{\square B}$. Hence $\dow\Omega=\Omega'$ and the right premise of cut is derived from $\lef A, \Omega',[\rig B]$.  We have the following sub-cases according to the shape of $A$, 
            \begin{itemize}
                \item $A=C\vee D$ and $r_1=\rig \vee$, 
                in this case, the left premise of $(\cut)$ is $                  \Omega_1,\rig {C\vee D}$ and is derived from $\Omega_1,\rig C$, 
                the right premise of $(\cut)$ is $\lef{C\vee D},\Omega_2',\rig{\square B}$. By Lemma \ref{inver-lemma}, both $\lef{C},\Omega_2',\rig{\square B}$ and $\lef{D},\Omega_2',\rig{\square B}$ are derivable at the same height. 
                Then the proof proceeds in the same way as in (D2).  $A=C\wedge D$ is similar. 
                \item $A=C\supset D$ and $r_1=\rig {\supset_1}$, we 
                apply $(\cut)$ to the premise of $(\rig{\square_1})$ and then apply $\rig{\square_1}$. 
                \item $A=\square C$ and $r_1=\rig{\square_1}$, 
                we have 
                \begin{center}
                \vspace{-.5cm}
                \begin{adjustbox}{max width=.85\textwidth}
                \begin{tabular}{ccc}
                $
                \vlderivation{
                \vliin{}{(\cut)}{
                    \Omega',\rig{\square B}
                }{
                    \vlin{(r_1)}{}{
                        \Omega',\rig {\square C}
                    }{
                    \vltr {\+D_1'} {
                        \Omega',\mblock{\rig C}
                    }{
                    \vlhy {\quad}}
                    {
                    \vlhy {}}
                    {
                    \vlhy {\quad}} 
                    }
                }
                {
                    \vlin{}{(\rig{\square_1})}{
                        \lef {\square C},\Omega',\rig{\square B}
                    }{
                    \vltr {\+D_2'} {
                        \lef {\square C},{\Omega'},\mblock{\rig B}
                    }{
                    \vlhy {\quad}}
                    {
                    \vlhy {}}
                    {
                    \vlhy {\quad}} 
                    }
                }
                }
                $
                & $\leadsto$ & 
                    \AXC{$\Omega',\mblock{\rig C}$}
                    \AXC{$\Omega',\rig{\square C}$}
                    \AXC{$\Omega',\lef{\square C},\mblock{\lef C,\rig B}$}
                    \RightLabel{$(\cut)$}
                    \BIC{$\Omega',\mblock{\lef C,\rig B}$}
                    \RightLabel{$(\mcut)$}
                    \BIC{$\Omega',\mblock{\rig B}$}
                    \RightLabel{$(\rig{\square_1})$}
                    \UIC{$\Omega',\rig{\square B}$}
                    \DP
                \end{tabular}
                \end{adjustbox}
                \end{center}
                \item $A=\Diamond C, r_1=\rig{\Diamond_1}$ and $\Omega'=\Omega'',\mblock{\Phi}$, since $\Omega=\Omega',\rig{\square B}$, we have 
                $\dow{\Omega}=\Omega'',\mblock{\Phi}$. 
                \begin{center}
                \vspace{-.3cm}
                \begin{adjustbox}{max width=.85\textwidth}
                \begin{tabular}{ccc}
                $
                \vlderivation{
                \vliin{}{(\cut)}{
                    \Omega'',\mblock{\Phi},\rig{\square B}
                }{
                    \vlin{(r_1)}{}{
                        \Omega'',\rig {\Diamond C},\mblock{\Phi}
                    }{
                    \vltr {\+D_1'} {
                        \Omega'',\mblock{\Phi,\rig C}
                    }{
                    \vlhy {\quad}}
                    {
                    \vlhy {}}
                    {
                    \vlhy {\quad}} 
                    }
                }
                {
                    \vlin{}{(\rig{\square_1})}{
                        \lef {\Diamond C},\Omega'',\mblock{\Phi},\rig{\square B}
                    }{
                    \vltr {\+D_2'} {
                        \lef {\Diamond C},{\Omega''},\mblock{\Phi},\mblock{\rig B}
                    }{
                    \vlhy {\quad}}
                    {
                    \vlhy {}}
                    {
                    \vlhy {\quad}} 
                    }
                }
                }
                $
                & $\leadsto$ & 
                    \AXC{$\Omega'',\mblock{\Phi,\rig C}$}
                    \RightLabel{$(w^+\upblock)$}
                    \UIC{$\Omega'',\mblock{\Phi},\mblock{\Phi,\rig C}$}
                    \AXC{$\Omega'',\mblock{\Phi},\rig{\square B},\mblock{\lef C}$}
                    \RightLabel{$(\lef w)$}
                    \UIC{$\Omega'',\mblock{\Phi},\rig{\square B},\mblock{\lef C,\Phi}$}
                    \RightLabel{$(\mcut)$}
                    \BIC{$\Omega'',\mblock{\Phi},\rig{\square B},\mblock{\Phi}$}
                    \RightLabel{$(c\upblock^-)$}
                    \UIC{$\Omega'',\mblock{\Phi},\rig{\square B}$}
                    \DP
            \end{tabular}
            \end{adjustbox}
            \end{center}
            \end{itemize}

            \item $r_2$ is one of the non-invertible rules, according to the form of rules, we see $r_2\in\{\rig{\Diamond},\rig{\vee_1},\rig{\vee_2}\}$. 
            Among all these three cases can be done by switching the order of $(e_2)$ and $(\cut)$, we only show the case when $r_2=\rig{\vee_2}$ as an example. Assume $\Omega=\Omega',\mblock{\Phi,\rig{B\vee C}}$. We have 
            \begin{center}
                \vspace{-.5cm}
                \begin{adjustbox}{max width=.95\textwidth}
                \begin{tabular}{ccc}
                    $
                    \vlderivation{
                    \vliin{}{(\cut)}{
                        \Omega',\mblock{\Phi,\rig{B\vee C}}
                    }{
                        \vlin{(e_1)}{}{
                            \Omega',\mblock{\Phi},\rig A
                        }{
                        \vltr {\+E_1'} {
                            \vdots
                        }{
                        \vlhy {\quad}}
                        {
                        \vlhy {}}
                        {
                        \vlhy {\quad}} 
                        }
                    }
                    {
                        \vlin{}{(\rig{\vee_2})}{
                            \lef A, \Omega',\mblock{\Phi,\rig{B\vee C}}
                        }{
                        \vltr {\+E_2'} {
                            \lef A,\Omega',\mblock{\Phi,\rig{B}}
                        }{
                        \vlhy {\quad}}
                        {
                        \vlhy {}}
                        {
                        \vlhy {\quad}} 
                        }
                    }
                    }
                    $
                    & $\leadsto$ & 
                    \AXC{$\Omega',\mblock{\Phi},\rig A$}
                \AXC{$\lef A,\Omega',\mblock{\Phi,\rig{B}}$}
                \RightLabel{$(\cut)$}
                \BIC{$\Omega',\mblock{\Phi,\rig{B}}$}
                \RightLabel{$(\rig{\vee_2})$}
                \UIC{$\Omega',\mblock{\Phi,\rig{B\vee C}}$}
                \DP
                \end{tabular}
                \end{adjustbox}
                \end{center}
        \end{enumerate}

    \noindent \textbf{(D4)} cut formula $\rig A$ is \textit{side} in the left premise. 
    This further implies the last rule $(r_1)$ applied in $\+D_1$ is one of the invertible $\bullet$-rules since other rules do not have side $\circ$-formulas (outside a block) on the conclusion. The proof is the same as the \textit{side} case (i) covered in \textbf{(D3)}. 

    \noindent \textbf{(E1)} 
    one of the premises is an axiom, 
    similar to the corresponding cases in (D1);

    \noindent \textbf{(E2)} cut formula $A$ is principal in both premises. Note that $\Phi$ is a simple sequent, namely block-free, so $A$ cannot be of the form $\square B$ or $\Diamond B$ 
    under the assumption of (E2). 
    $A= B\vee C$ or $B\wedge C$ is trivial, we only consider $A= B\supset C$, we have
    
    \begin{center}
            \vspace{-.8cm}
            \begin{adjustbox}{max width=\textwidth}
            \begin{tabular}{c}
            $
            \vlderivation{
            \vliin{}{(\mcut)}{\Omega,\mblock{\Phi}}{
            \vlin{}{(\rig {\supset_2})}{
                \dow\Omega,\mblock{\dow\Phi,\rig{B\supset C}}
                }{
                    \vltr {\+E_1'} {
                        \dow\Phi,\lef B,\rig C
                        }{
                        \vlhy {\quad}}
                        {
                        \vlhy {}}
                        {
                        \vlhy {\quad}}
                }
            }
            {
                \vliin{}{(\lef {\supset_2})}{
                \Omega,\mblock{\lef{B\supset C},\Phi}
                }{
                    \vltr {\+E_2'} {
                        \dow\Omega,\mblock{\lef{B\supset C},\dow\Phi, \rig B}
                        }{
                        \vlhy {\quad}}
                        {
                        \vlhy {}}
                        {
                        \vlhy {\quad}}
                }
                {
                    \vltr {\+E_2''} {
                        \Omega,\mblock{\lef{C},\Phi}
                        }{
                        \vlhy {\quad}}
                        {
                        \vlhy {}}
                        {
                        \vlhy {\quad}}
                }
            }
            }
            $
            \quad $\leadsto$\\ \\
            \AXC{$\dow\Omega,\mblock{\dow\Phi,\rig{B\supset C}}$}
            \AXC{$\dow\Omega,\mblock{\lef{B\supset C},\dow\Phi,\rig B}$}
            \RightLabel{$(\mcut)$}
            \BIC{$\dow\Omega,\mblock{\dow\Phi,\rig B}$}
            \AXC{$\Phi^*,\lef B,\rig C$}
            \RightLabel{$(\mblock{\cdot})$}
            \UIC{$\dow\Omega,\mblock{\Phi^*,\lef B,\rig C}$}
            \RightLabel{$(\mcut)$}
            \BIC{$\dow\Omega,\mblock{\Phi^*,\rig C}$}
            \AXC{$\Omega,\mblock{\lef{C},\Phi}$}
            \RightLabel{$(\mcut)$}
            \BIC{$\Omega,\mblock{\Phi}$}
            \DP
            \end{tabular}
        \end{adjustbox}
        \end{center}

    \noindent \textbf{(E3)} 
    $\rig A$ is principal in $(e_1)$ while $\lef A$ is \textit{side} in $(e_2)$. Similar to \textbf{(D3)}.

        \noindent \textbf{(E4)} $\rig A$ is \textit{side} in $(e_1)$. This implies the last rule $(e_1)$ applied in $\+E_1$ is one of the $\bullet$-invertible rules or is $(\lef{\Diamond_2})$. For the invertible rules, the proof is similar to the \textit{side} case (i) covered in \textbf{(D3)}. Here we only show the case when $e_1=\lef{\Diamond_2}$. In this case, $\Phi=\lef{\Diamond B}, \Phi'$
        and derivation $\+E$ is unfolded as
        \vspace{-.5cm}
        \begin{center}
            \AXC{$\vdots$}
            \LeftLabel{$e_1'$}
            \UIC{${\Phi'}^*,\mblock{\lef B},\rig A$}
            \LeftLabel{$e_1$}
            \UIC{$\dow\Omega,\mblock{{\Phi'}^*,\lef{\Diamond B}, \rig A}$}
            \AXC{$\vdots$}
            \LeftLabel{$e_2$}
            \UIC{$\Omega,\mblock{\lef A,\Phi',\lef{\Diamond B}}$}
            \BIC{$\Omega,\mblock{\Phi',\lef{\Diamond B}}$}
            \DP
        \end{center}
        We discuss the role of $\lef A$ in $(e_2)$ and the role of $\rig A$ in $(e_1')$. We have the following cases where neither $\lef A$ nor $\rig A$ is weak:
        \begin{enumerate}[label=(\roman*)]
            \item both $\lef A$ and $\rig A$ are \textit{principal}, then according to the form of rules, $A$ cannot be a $\square$-formula. Thus we have the following sub-cases:
            \begin{itemize}
                \item $A=C\wedge D$ or $C\vee D$. We only show the first one. In this case, we have 
                \begin{center}
                    \vspace{-.5cm}
                    \begin{adjustbox}{max width=.9\textwidth}
                    \begin{tabular}{c}                        
                    $
                    \vlderivation{
                    \vliin{}{(\mcut)}{
                        \Omega,\mblock{\Phi',\lef{\Diamond B}}
                    }{
                        \vlin{(e_1)}{}{
                            \dow\Omega,\mblock{\dow{\Phi'},\lef{\Diamond B},\rig {C\wedge D}}
                        }{
                        \vliin{(e_1')}{}{
                            \dow{\Phi'},\mblock{\lef B},\rig {C\wedge D}
                            }
                        {
                        \vltr {\+E_{11}'} {
                        \dow{\Phi'},\mblock{\lef B},\rig {C}
                        }{
                        \vlhy {\quad}}
                        {
                        \vlhy {}}
                        {
                        \vlhy {\quad}} 
                        }
                        {
                        \vltr {\+E_{12}'} {
                                \dow{\Phi'},\mblock{\lef B},\rig {D}
                                }{
                                \vlhy {\quad}}
                                {
                                \vlhy {}}
                                {
                                \vlhy {\quad}}    
                        }
                        }
                    }
                    {
                        \vlin{}{(e_2)}{
                            \Omega,\mblock{\lef {C\wedge D},\Phi',\lef{\Diamond B}}
                        }{
                        \vltr{\+E_2'} {
                            \Omega,\mblock{\lef {C},\lef D,\Phi',\lef{\Diamond B}}    
                                    }{
                                    \vlhy {\quad}}
                                    {
                                    \vlhy {}}
                                    {
                        \vlhy {\quad}} 
                        }
                    }
                    }
                    $
                    \quad $\leadsto$\\ \\
                    \AXC{$\dow{\Phi'},\mblock{\lef B},\rig {D}$}
                    \RightLabel{$(\lef{\Diamond_2})$}
                    \UIC{$\dow\Omega,\mblock{\dow{\Phi'},\lef{\Diamond B},\rig {D}}$}
                    \AXC{$\dow{\Phi'},\mblock{\lef B},\rig {C}$}
                    \RightLabel{$(\lef{\Diamond_2})$}
                    \UIC{$\dow\Omega,\mblock{\dow{\Phi'},\lef{\Diamond B},\rig {C}}$}
                    \RightLabel{$(\lef w)$}
                    \UIC{$\dow\Omega,\mblock{\dow{\Phi'},\lef{\Diamond B},\lef D,\rig {C}}$}
                    \AXC{$\Omega,\mblock{\lef {C},\lef D,\Phi',\lef{\Diamond B}}   $}
                    \RightLabel{$(\mcut)$}
                    \BIC{$\Omega,\mblock{\lef D,\Phi',\lef{\Diamond B}}$}
                    \RightLabel{$(\mcut)$}
                    \BIC{$\Omega,\mblock{\Phi',\lef{\Diamond B}}$}
                    \DP
                \end{tabular}
            \end{adjustbox}
            \end{center} 
                \item $A=C\supset D$. In this case, derivation $\+E$ is unfolded as
                \begin{center}
                    \vspace{-.5cm}
                    $
                    \vlderivation{
                    \vliin{}{(\mcut)}{
                        \Omega,\mblock{\Phi',\lef{\Diamond B}}
                    }{
                        \vlin{(e_1)}{}{
                            \dow\Omega,\mblock{\dow{\Phi'},\lef{\Diamond B},\rig {C\supset D}}
                        }{
                        \vlin{(e_1')}{}{
                            \dow{\Phi'},\mblock{\lef B},\rig {C\supset D}
                            }
                        {
                        \vltr {\+E_{1}'} {
                        \dow{\Phi'},\mblock{\lef B},\lef C,\rig {D}
                        }{
                        \vlhy {\quad}}
                        {
                        \vlhy {}}
                        {
                        \vlhy {\quad}} 
                        }
                        }
                    }
                    {
                        \vliin{}{(e_2)}{
                            \Omega,\mblock{\lef {C\supset D},\Phi',\lef{\Diamond B}}
                        }{
                        \vltr{\+E_{21}'} {
                            \dow\Omega,\mblock{\lef {C\supset D},\dow{\Phi'},\lef{\Diamond B},\rig C}    
                                    }{
                                    \vlhy {\quad}}
                                    {
                                    \vlhy {}}
                                    {
                        \vlhy {\quad}} 
                        }
                        {
                            \vltr{\+E_{22}'} {
                                \Omega,\mblock{\lef D,\Phi',\lef{\Diamond B}}    
                                        }{
                                        \vlhy {\quad}}
                                        {
                                        \vlhy {}}
                                        {
                            \vlhy {\quad}}    
                        }
                    }
                    }
                    $
                \end{center}
                and we transform it into
                \begin{center}
                \begin{adjustbox}{max width=.85\textwidth}
                    \AXC{$\dow\Omega,\mblock{\dow{\Phi'},\lef{\Diamond B},\rig {C\supset D}}$}
                    \AXC{$\dow\Omega,\mblock{\lef {C\supset D},\dow{\Phi'},\lef{\Diamond B},\rig C}$}
                    \RightLabel{$(\mcut)$}
                    \BIC{$\dow\Omega,\mblock{\dow{\Phi'},\lef{\Diamond B},\rig C}$}
                    \AXC{$\dow{\Phi'},\mblock{\lef B},\lef C,\rig {D}$}
                    \RightLabel{$(\lef{\Diamond_2})$}
                    \UIC{$\dow\Omega,\mblock{\dow{\Phi'},\lef{\Diamond B},\lef C,\rig { D}}$}
                    \RightLabel{$(\mcut)$}
                    \BIC{$\dow\Omega,\mblock{\dow{\Phi'},\lef{\Diamond B},\rig {D}}$}
                    \AXC{$\Omega,\mblock{\lef D,\Phi',\lef{\Diamond B}}$}
                    \RightLabel{$(\mcut)$}
                    \BIC{$\Omega,\mblock{{\Phi'},\lef{\Diamond B}}$}
                    \DP
                \end{adjustbox}
                \end{center}
                \item $A=\Diamond C$. In this case, we have 
                \begin{center}
                    \vspace{-.3cm}
                    \begin{adjustbox}{max width=.85\textwidth}
                    \begin{tabular}{ccc}                        
                    $
                    \vlderivation{
                    \vliin{}{(\mcut)}{
                        \Omega,\mblock{\Phi',\lef{\Diamond B}}
                    }{
                        \vlin{(e_1)}{}{
                            \dow\Omega,\mblock{\dow{\Phi'},\lef{\Diamond B},\rig {\Diamond C}}
                        }{
                        \vlin{(e_1')}{}{
                            \dow{\Phi'},\mblock{\lef B},\rig {\Diamond C}
                            }
                        {
                        \vltr {\+E_{1}'} {
                        \dow{\Phi'},\mblock{\lef B,\rig C}
                        }{
                        \vlhy {\quad}}
                        {
                        \vlhy {}}
                        {
                        \vlhy {\quad}} 
                        }
                        }
                    }
                    {
                        \vlin{}{(e_2)}{
                            \Omega,\mblock{\lef {\Diamond C},\Phi',\lef{\Diamond B}}
                        }{
                        \vltr{\+E_{2}'} {
                            \mblock{\lef C}, {\Phi'}^+, \lef{\Diamond B}    
                                    }{
                                    \vlhy {\quad}}
                                    {
                                    \vlhy {}}
                                    {
                        \vlhy {\quad}} 
                        }
                    }
                    }
                    $
                    & $\leadsto$ & 
                    \AXC{$\dow{\Phi'},\mblock{\lef B,\rig C}$}
                    \RightLabel{$(\lef w)$}
                    \UIC{$\dow{\Phi'},\mblock{\lef B,\rig C},\lef{\Diamond B}$}
                    \AXC{$\mblock{\lef C}, {\Phi'}^+, \lef{\Diamond B}$}
                    \RightLabel{$(\lef w)$}
                    \UIC{$\mblock{\lef C,\lef B}, {\Phi'}^+, \lef{\Diamond B}$}
                    \RightLabel{$(\mcut)$}
                    \BIC{$\mblock{\lef B}, \Phi'^+, \lef{\Diamond B}$}
                    \RightLabel{$(\lef{\Diamond_2})$}
                    \UIC{$\Omega,\mblock{\lef{\Diamond B}, \Phi',\lef{\Diamond B}}$}
                    \RightLabel{$(\lef c)$}
                    \UIC{$\Omega,\mblock{\Phi',\lef{\Diamond B}}$}
                    \DP
                    \end{tabular}
                \end{adjustbox}
            \end{center}
            \end{itemize} 
            \item $\lef A$ is \textit{side} in $(e_2)$, then according to the form of rules, $e_2$ can be any rule of the calculus. 
            \begin{itemize}
                \item if $e_2$ is an invertible rule or $\rig{\Diamond},\rig{\vee}$, we switch the order of $(e_2)$ and $(\mcut)$, that is, apply $(\mcut)$ to the premise of $(e_2)$ and the conclusion of $(e_1')$ first, then apply $(e_2)$. 
                \item if $e_2\in\{\rig{\supset_2},\rig{\square_2},\lef{\Diamond_2}\}$, the three cases are similar and we only show one here. Assume $\Phi'=\Psi,\rig{C\supset D}$, then $\Omega^*=\Omega$, $\dow{\Phi'}=\Psi$, 
                and we have 
                \begin{center}
                    \vspace{-.8cm}
                    \begin{adjustbox}{max width=.9\textwidth}
                    \begin{tabular}{ccc}                        
                    $
                    \vlderivation{
                    \vliin{}{(\mcut)}{
                        \Omega,\mblock{\Psi,\lef{\Diamond B},\rig{C\supset D}}
                    }{
                        \vlin{(e_1)}{}{
                            \Omega,\mblock{\Psi,\lef{\Diamond B},\rig A}
                        }{
                        \vlin{(e_1')}{}{\Psi,\mblock{\lef B},\rig A}
                        {
                        \vltr {\+E_1'} {
                        \vdots
                        }{
                        \vlhy {\quad}}
                        {
                        \vlhy {}}
                        {
                        \vlhy {\quad}} 
                        }
                        }
                    }
                    {
                        \vlin{}{(e_2)}{
                            \Omega,\mblock{\lef A,\Psi,\lef{\Diamond B},\rig{C\supset D}}
                        }{
                        \vltr{\+E_2'} {
                            \lef A,\Psi,\lef{\Diamond B},\lef C, \rig{D}    
                                    }{
                                    \vlhy {\quad}}
                                    {
                                    \vlhy {}}
                                    {
                        \vlhy {\quad}} 
                    }
                    }
                    }
                    $
                    & $\leadsto$ & 
                    \AXC{$\Psi,\mblock{\lef B},\rig A$}
                    \RightLabel{$(\lef w)$}
                    \UIC{$\Psi,\mblock{\lef B},\lef{\Diamond B}, \lef C, \rig A$}
                    \AXC{$\lef A,\Psi,\lef{\Diamond B},\lef C, \rig{D}$}
                    \RightLabel{$(w^+\upblock)$}
                    \UIC{$\lef A,\Psi,\mblock{\lef B},\lef{\Diamond B},\lef C, \rig{D}$}
                    \RightLabel{$(\cut)$}
                    \BIC{$\Psi,\mblock{\lef B},\lef{\Diamond B},\lef C, \rig{D}$}
                    \RightLabel{$(\lef{\Diamond_1})$}
                    \UIC{$\Psi,\lef{\Diamond B},\lef{\Diamond B},\lef C, \rig{D}$}
                    \RightLabel{$(\lef c)$}
                    \UIC{$\Psi,\lef{\Diamond B},\lef C, \rig{D}$}
                    \RightLabel{$(\rig{\supset_2})$}
                    \UIC{$\Omega,\mblock{\Psi,\lef{\Diamond B},\rig{C\supset D}}$}
                    \DP
                \end{tabular}
                \end{adjustbox}
                \end{center}
            \end{itemize}
            \item $\rig A$ is \emph{side} in $(e_1')$ and $\lef A$ is \textit{principal} in $(e_2)$. Similar to (iii), $A$ cannot be a $\square$-formula. For cases when $A=C\wedge D$ or $C\supset D$, the proof is the same to (iii), since the corresponding $\circ$-rules with respect to $\rig A$ are invertible. Thus we only have two cases remaining: $A=\Diamond D$ or $A=C\vee D$. We show the former case as the latter is similar. 
                
            In this case, the derivation $\+E$ is unfolded as follows, where $\Psi=\Phi'$, 
            \begin{center}
            \vspace{-.5cm}
            \AXC{$\vdots$}
            \LeftLabel{$e_1'$}
            \UIC{$\dow\Psi,\mblock{\lef B},\rig{\Diamond D}$}
            \LeftLabel{$e_1$}
            \UIC{$\dow\Omega,\mblock{\dow\Psi,\lef{\Diamond B}, \rig{\Diamond D}}$}
            \AXC{$\vdots$}
            \UIC{$\mblock{\lef D},\Psi^+,\lef{\Diamond B}$}
            \RightLabel{$e_2$}
            \UIC{$\Omega,\mblock{\lef{\Diamond D},\Psi,\lef{\Diamond B}}$}
            \RightLabel{$\mcut$}
            \BIC{$\Omega,\mblock{\Psi,\lef{\Diamond B}}$}
            \DP        
            \end{center}
                    
                    \noindent (a) If $e_1'=\lef{\Diamond_2}$, the conclusion of $\cut$ follows from the premise of $(e_1')$ by $(\lef{\Diamond_2})$ twice. 
                    
                    \noindent (b) If $e_1'$ is one of the $\bullet_1$-rules except $\lef{\Diamond_1}$, then 
                    from the premise of $(e_1')$, we apply $(e_1)$ first, then $(\mcut)$ to the conclusion of $(e_1)$ with the conclusion of $(e_2)$, and the $\bullet_2$-counterpart of $(e_1')$ in the end. 
                    
                    \noindent 
                    (c) If $e_1'$ is $\rig{\Diamond}$, then $\Psi,\mblock{\lef B,\rig D}$ is derivable, we transform the whole derivation into
                    \vspace{-.3cm}
                    \begin{prooftree}
                        \AXC{$\Psi,\mblock{\lef B,\rig D}$}
                        \RightLabel{$(\lef w)$}
                        \UIC{$\Psi,\lef{\Diamond B}, \mblock{\lef B,\rig D}$}
                        \AXC{$\mblock{\lef D}, \Psi^+, \lef{\Diamond B}$}
                        \RightLabel{$(\lef w)$}
                        \UIC{$\mblock{\lef D,\lef B}, \Psi^+, \lef{\Diamond B}$}
                        \RightLabel{$(\mcut)$}
                        \BIC{$\mblock{\lef B}, \Psi^+, \lef{\Diamond B}$}
                        \RightLabel{$(\lef{\Diamond_2})$}
                        \UIC{$\Omega,\mblock{\Psi,\lef{\Diamond B},\lef{\Diamond B}}$}
                        \RightLabel{$(\lef{c})$}
                        \UIC{$\Omega,\mblock{\Psi,\lef{\Diamond B}}$}
                    \end{prooftree}
                    
                    \noindent (d) If $e_1'$ is one of the $\bullet_2$-rules or $\lef{\Diamond_1}$, let us consider the rules applied above $e_1'$. Assume that before any $\circ$-rule is applied, there are no $\bullet_1$-rules applied above $e_1'$, 
                    otherwise we return to case (a), switching the rule order and simulating the derivation with the $\bullet_2$-counterpart as above. Let us enumerate the branches of $\+E_1'$ as $\+B_1,\ldots,\+B_k$ and for each $\+B_i$ further enumerate sequents contained in it as $\Phi_{i_0},\ldots,\Phi_{i_j},\ldots,\Phi_{i_m}$ such that $\Phi_{i_0}=\dow{\Phi'},\mblock{\lef B},\rig {\Diamond D}$ and $\Phi_{i_m}$ is an axiom.   
                    Then for each $\+B_i$, it implies one of the following situations regarding the output formula: 
                    \begin{enumerate}
                        \item[(h1)] $\rig{\Diamond D}$ is in each $\Phi_{i_j}$. 
                        \item[(h2)] there is some $i_k\geq 1$ such that 
                        $\Phi_{i_k}=\Theta,\mblock{\Lambda},\rig{\Diamond D}$ for some $\Theta,\Lambda$ and $\Phi_{i_{k+1}}=\Theta,\mblock{\Lambda,\rig D}$; meanwhile, for any $i_j\leq i_k$, the output formula of $\Phi_{i_j}$ is $\rig{\Diamond D}$. 
                        \item[(h3)] there is some $i_k\geq 1$ such that $\Phi_{i_k}=\Theta,\mblock{\Psi,\lef{P\supset Q}},\rig{\Diamond D}$ for some $\Theta,\Psi,P,Q$ and $\Phi_{i_{k+1}}=\Theta,\mblock{\Psi,\lef{P\supset Q},\rig P}$; meanwhile, for any $i_j\leq i_k$, the output formula of $\Phi_{i_j}$ is $\rig{\Diamond D}$. 
                    \end{enumerate}
                    We process each $\+B_i$ respectively and construct the  corresponding $\+B_i'$ for it. 
                    If (h1) holds, this means $\rig{\Diamond D}$ is unused and hence side in the whole branch, so we can just replace each occurrence of $\rig D$ in $\+B_i$ with the output formula $\rig F$ of $\Psi^+$ 
                    and then get $\+B_i'$. 
                    If (h2) holds, we reconstruct $\+B_i$ into $\+B_i'$ as follows. 
                    First we keep all the components $\Phi_{i_j}$'s for $j>k+1$. 
                    Next, we remove $\Phi_{i_k}$; and as in (b), cut $\Phi_{i_{k+1}}$ with the premise of $(e_2)$ and obtain $\Phi_{i_k}'$ as follows
                    \vspace{-.2cm}
                    \begin{prooftree}
                        \AXC{$\Theta,\mblock{\Lambda,\rig D}$}
                        \RightLabel{$(\lef w)$}
                        \UIC{$\Psi,\lef{\Diamond B}, \Theta, \mblock{\Lambda,\rig D}$}
                        \AXC{$\mblock{\lef D}, \Psi^+, \lef{\Diamond B}$}
                        \RightLabel{$(\lef w)$}
                        \UIC{$\mblock{\Lambda,\lef D}, \Psi^+, \lef{\Diamond B},\Theta$}
                        \RightLabel{$(\mcut)$}
                        \BIC{$\Psi^+, \lef{\Diamond B},\Theta,\mblock{\Lambda}$}
                    \end{prooftree}
                    Then for all the $\Phi_{i_j}$'s with $j<k$, let $\Phi'_{i_j}=\Psi,\lef{\Diamond B}, \dow\Phi_{i_j},\rig F$. 
                    
                    Lastly, if (h3) holds, to construct $\+B_i'$ from the current $\+B$, we keep all the components $\Phi_{i_j}$'s for $j>k$ and for $j\leq k$, we replace $\rig {\Diamond D}$ with $\rig F$, which is the output formula of the final conclusion of $(\mcut)$. Thus, for instance, $\Phi_{i_k}'=\Theta,[\Psi,\lef{P\supset Q}],\rig F$. 
                    Note that the way we process each branch $\+B$ in fact 
                    eliminates the cut formula $\rig{\Diamond D}$ in the whole branch by substituting with the output formula at certain points, 
                    but this does not influence other rule applications in the derivation, thus we can connect each $\+B_i'$ and reconstruct the derivation with contraction ending with $\Psi^+,\lef{\Diamond B},[\lef B]$. Then by $(\lef {\Diamond_2})$ and contraction, we conclude $\Omega,[\Psi,\lef{\Diamond B}]$. 
        \end{enumerate}
    This completes the whole proof of cut-admissibility. 
\end{proof}

By the previous theorem we finally obtain the completeness of \lcfik. 

\begin{theorem}[Completeness of \lcfik]
    If a formula $A$ is provable in $\hfik$, then it is provable in \lcfik. 
\end{theorem}

\section{Proof search and complexity}\label{sec:termination-complexity}
In this section we show that the decision problem of \fik is in \textsc{Expspace} by means of proof search in the calculus \lcfik with some refinement. Before presenting the main result, we can easily 
show the calculus \lcfik provides a decision problem with some rough upper bound by just minimal means. 

To obtain decidability and a loose upper bound, we proceed as follows. We know by proposition \ref{prop:admissibility-weakening-contraction} that contraction and weakening are admissible, hence we can restrict our attention to `set-sequents' in the following sense: 
\begin{definition}[(quasi-)set-sequents]
    Let $\Omega$ be a sequent of the form $\Phi, [\Psi_1], \ldots, [\Psi_1]$. 
    $\Omega$ is called a \emph{quasi-set-sequent} 
    if each $\Phi, \Psi_i$ is a set (rather than a multiset) of formulas; 
    $\Omega$ is further called a \emph{set-sequent} if it is a quasi-set-sequent 
    and for $i\not= j$, $\Psi_i \not= \Psi_j$. 
\end{definition}
Set-sequents in this sense do not contain duplicated formulas in the same `context' (block/top-level sequent), nor duplicated blocks. 
It is then easy to prove that 
given a formula $\rig A$ (as a root of a derivation), there are only finitely many \emph{distinct} set-sequents that can occur in any derivation of $A$, no matter what are the rules of the calculus: 

\begin{proposition}\label{prop:max-block}
	Let $A$ be a formula and $|A|=\+O(n)$. Then the size of any  set-based sequent that may occur in any possible derivation of $A$ is $2^{\+O(n)}$, whence the set of all possible sequents that can occur in any possible derivation has cardinality $2^{2^{\+O(n)}}$. 
\end{proposition}

\begin{proof}
	Each set-sequent $\Omega$ is a member of the set $\+P(\subfm{A})\times \+P(\+P(\subfm{A}))$ 
	Let $|A|=\+O(n)$, we have $|\subfm{A}|=\+O(n)$, whence the size of $\Omega$ is $2^{\+O(n)}$  and therefore the cardinality of the sets of such sequents is the cardinality of the set $\+P(\subfm{A})\times \+P(\+P(\subfm{A}))$ which is $2^{2^{\+O(n)}}$. 
\end{proof}

From this result we can  obtain directly a decision procedure and an upper bound for it: given a formula $\rig A$ as root,  we carry on backward proof search in \lcfik storing one branch of a derivation at a time with the following provisos: 
(i) we do not apply rules to an axiom, (ii) we remove duplicated formulas or blocks in order to keep the set-based structure, (iii) we stop proof-search if the same sequent already appears in the branch. By the previous result it follows that the size of each branch is bounded by $2^{2^{\+O(n)}}$. Whence the decision procedure runs in \textsc{2-Expspace}. 
Notice that this bound, although elevated is elementary whence already much lower than what is conjectured for \ik. 

The rest of the section is devoted to refine this upper bound to \textsc{Expspace}. 
Although as many as $2^{2^{\+O(n)}}$ set sequents can theoretically occur in a proof branch (by Proposition \ref{prop:max-block}), we aim to show that only $2^{\+O(n)}$ of them can actually occur.
Our plan is as follows: 
We first define  a variant of  \lcfik called \cumlcfik based on quasi-set-sequents  that  preserves the admissibility of structural rules. 
Then we show that if a quasi-set-sequent is provable in \lcfik then it is provable in \cumlcfik. 
Next we refine the notion of set-sequents to \emph{tight sequents} which are set-sequents with an additional restriction.  
We show if a formula is provable in \cumlcfik then it has a \emph{clean} proof where any sequent is a tight sequent. Finally we define a suitable proof search strategy for finding clean proofs for formulas and show that the proof search procedure runs in \textsc{Expspace}. 

We define the calculus \cumlcfik for \fik based on quasi-set-sequents by modifying some of the rules in \lcfik. The modified rules are found in Figure \ref{fig:shallow-fik-cumulative} and for other rules such as most $\circ$-rules except $\rig\supset$ as well as $\lef\square$, they keep the same form as in \lcfik. 
Comparing with \lcfik there are a few changes. 
First, the input part of each rule now becomes \emph{locally} cumulative by repeating the principal formula (if it is a $\bullet$-formula) with respect to the corresponding block/top-level sequent. Second, as part of the machinery to avoid redundant rule applications, we divide both $\rig\supset_1$ and $\rig\supset_2$ into two sub-rules according to whether the premise $A$ of the principal formula $\rig {A\supset B}$ already occurs on formula part of the conclusion. 

\begin{figure}[!t]
    \begin{framed}
    \begin{center}
        \begin{adjustbox}{max width = \textwidth}
            \begin{tabular}{cc}
                \AxiomC{$\Omega,\lef {A\wedge B},\lef A,\lef B$}
                \RightLabel{($\lef \wedge_1$)}
                \UnaryInfC{$\Omega,\lef {A\wedge B}$}
                \DisplayProof
                &
                \AxiomC{$\Omega,\mblock{\Phi,\lef A,\lef B,\lef {A\wedge B}}$}
                \RightLabel{($\lef \wedge_2$)}
                \UnaryInfC{$\Omega,\mblock{\Phi,\lef {A\wedge B}}$}
                \DisplayProof
                \\[.4cm]
                \AxiomC{$\Omega,\lef A,\lef{A\vee B}$}
                \AxiomC{$\Omega,\lef B,\lef{A\vee B}$}
                \RightLabel{($\lef \vee_1$)}
                \BinaryInfC{$\Omega,\lef{A\vee B}$}
                \DisplayProof
                &
                \AxiomC{$\Omega,\mblock{\Phi,\lef A,\lef{A\vee B}}$}
                \AxiomC{$\Omega,\mblock{\Phi,\lef B,\lef{A\vee B}}$}
                \RightLabel{($\lef \vee_2$)}
                \BinaryInfC{$\Omega,\mblock{\Phi,\lef{A\vee B}}$}
                \DisplayProof
                \\[.4cm]
                \AxiomC{$\Omega^*,\lef{A\supset B},\rig A$}
                \AxiomC{$\Omega,\lef B,\lef{A\supset B}$}
                \RightLabel{($\lef \supset_1$)}
                \BinaryInfC{$\Omega, \lef{A\supset B}$}
                \DisplayProof
                &
                \AxiomC{$\Omega^*,\mblock{\dow\Phi,\lef{A\supset B},\rig A}$}
                \AxiomC{$\Omega,\mblock{\Phi,\lef B,\lef{A\supset B}}$}
                \RightLabel{($\lef \supset_2$)}
                \BinaryInfC{$\Omega,\mblock{\Phi, \lef{A\supset B}}$}
                \DisplayProof
                \\[.4cm]
                \AXC{$\Omega,\lef A,\rig B$}
                \RightLabel{$(\rig \supset_1), A\notin\fm{\Omega}$}
                \UIC{$\Omega,\rig{A\supset B}$}
                \DP
                \quad
                \AXC{$\Omega,\rig B$}
                \RightLabel{$(\rig \supset_1), A\in\fm{\Omega}$}
                \UIC{$\Omega,\rig{A\supset B}$}
                \DP
                &
                \AXC{$\Phi,\lef A,\rig B$}
                \RightLabel{$(\rig \supset_2), A\notin \Phi$}
                \UIC{$\Omega,\mblock{\Phi,\rig{A\supset B}}$}
                \DP
                \quad
                \AXC{$\Phi,\rig B$}
                \RightLabel{$(\rig \supset_2), A\in \Phi$}
                \UIC{$\Omega,\mblock{\Phi,\rig{A\supset B}}$}
                \DP
                \\[.4cm]
                \AXC{$\Omega,\lef {\Diamond A},\mblock{\lef A}$}
                \RightLabel{$(\lef \Diamond_1)$}
                \UIC{$\Omega,\lef {\Diamond A}$}
                \DP
                &
                \AXC{$\Phi^+,\lef {\Diamond A},\mblock{\lef A}$}
                \RightLabel{$(\lef \Diamond_2)$}
                \UIC{$\Omega,\mblock{\Phi,\lef {\Diamond A}}$}
                \DP
            \end{tabular}
        \end{adjustbox}
    \end{center}
    \caption{\cumlcfik, based on quasi-set-sequents}
    \label{fig:shallow-fik-cumulative}
    \end{framed}
\end{figure}

The following proposition is obtained directly by contraction in \lcfik: 

\begin{proposition}
    If a quasi-set-sequent $\Omega$ is provable in \lcfik then it is provable in \cumlcfik. 
\end{proposition}

Note that \cumlcfik operates on quasi-set-sequents for which we only assume internal contraction (contraction of formulas within a block) but allow duplicated blocks, hence weakening and external contraction (contraction of blocks) are hp-admissible as in the original \lcfik. 
Therefore, we have 

\begin{proposition}\label{prop:tight-preserving}
    Let $\Phi_1\subseteq\Phi_2$ be two sets of input formulas. 
    The following hold for quasi-set-sequents

    \noindent if \ $\vdash_{\mcumlcfik}^n \Omega,\mblock{\Phi_1},\rig F$ then $\vdash_{\mcumlcfik}^n \Omega,\mblock{\Phi_1},\mblock{\Phi_2}, \rig F$ \qquad  if \ $\vdash_{\mcumlcfik}^n \Omega,\mblock{\Phi_1,\rig F}$ then $\vdash_{\mcumlcfik}^n \Omega,\mblock{\Phi_1},\mblock{\Phi_2,\rig F}$
\end{proposition}

Next, we introduce a more restricted notion of set-sequents called `tight sequent', 
which is required to generate derivations in which each sequent is set-based and at the same time its input parts are non-decreasing, meaning that the $\bullet$-formulas are the same or cumulating along derivations, with respect to the same block. 
This in turn is essential to get the refined complexity bound. 

\begin{definition}[tight sequent and clean derivation]
A \emph{tight} sequent is a set-sequent satisfying the additional condition: if there is a block containing   the output formula, then the input part of that block is either empty or is not 
identical to 
any other block. A derivation (resp. proof) containing only tight sequents is called \emph{clean} derivation (resp. proof). 
\end{definition}

\begin{example}
    $\lef {p\vee q}, \mblock{\lef p, \lef q},\mblock{\lef q}, \rig p$ is a tight sequent while neither (1) $\lef {p\vee q}, \mblock{\lef p, \lef q},\mblock{\lef q, \lef p, \rig p}$ nor (2) $\lef {p\vee q}, \mblock{\lef p, \lef q},\mblock{\lef p,\lef q}, \rig p$ is, since in (1) the input of the block containing output formula is indentical to another block; in (2) there are duplicated blocks. 
\end{example}

\begin{restatable}{propp}{cleanproof}\label{prop:clean-proof}
    If a tight sequent $\Omega,\mblock{\Phi_1},\ldots,\mblock{\Phi_k}$ (possibly $k=0$) is provable in \cumlcfik then 
    there exist $\Psi_1,\ldots,\Psi_l$ such that $\bigcup_{i=1}^l\Psi_i \subseteq \bigcup_{j=1}^k\Phi_j$, $l\leq k$
    and the tight sequent $\Omega,\mblock{\Psi_1},\ldots,\mblock{\Psi_{l}}$ 
    has a clean proof. 
\end{restatable}

\begin{proof}
    Let $\Sigma=\Omega,\mblock{\Phi_1},\ldots,\mblock{\Phi_k}$ and call sequents $\Omega,\mblock{\Psi_1},\ldots,\mblock{\Psi_{l}}$ of the form in the proposition its sub-sequents. 
    Assume $\Sigma$ is provable in \cumlcfik with a proof $\+D$, we show that 
    $\+D$ can be transformed into a clean proof for some subsequent $\Sigma_1$. The idea is, we find the first application which produces non-tight sequents in the premise and then reconstruct the derivation. We do this iteratively until we get a clean proof of some $\Sigma_1$. 
    First if a tight sequent is an initial sequent in \lcfik, then itself is a clean proof. Otherwise, let us consider the first rule application $(r)$ that might produce non-tight sequents. We show the following cases and others are similar. 

    \noindent (1) $r=\lef {\wedge_2}$ or $r=\lef {\vee_2}$, we only show the case of conjunction. 
    Consider an application with premise $\Omega',\mblock{\Phi,\lef{A\wedge B}, \lef A,\lef B}, \mblock{\Phi,\lef{A\wedge B},\lef A,\lef B}$ and conclusion $\Omega',\mblock{\Phi,\lef{A\wedge B}}, \mblock{\Phi,\lef{A\wedge B},\lef A,\lef B}$. 
    From the premise by contraction in \cumlcfik, we obtain $\Omega',\mblock{\Phi,\lef{A\wedge B}, \lef A,\lef B}$ which is a tight sequent. We remove this rule application in $\+D$ and reconstruct $\+D$ as follows: 
    For rule applications above the original premise $\Omega',\mblock{\Phi,\lef{A\wedge B}, \lef A,\lef B}, \mblock{\Phi,\lef{A\wedge B},\lef A,\lef B}$, if it happens only in $\Omega'$, we keep as it is; if it involves any of the two blocks $\mblock{\Phi,\lef{A\wedge B}, \lef A,\lef B}, \mblock{\Phi,\lef{A\wedge B},\lef A,\lef B}$, we apply them only in the single one. 
    For rule applications below the original conclusion $\Omega',\mblock{\Phi,\lef{A\wedge B}}, \mblock{\Phi,\lef{A\wedge B},\lef A,\lef B}$, we remove them if it involves the block $\mblock{\Phi,\lef{A\wedge B}}$ and keep the rest. 

    \noindent (2) $r=\rig {\Diamond}$. Consider an application with premise $\Omega',\mblock{\Phi,\rig A}, \mblock{\Phi,\Theta}$ and conclusion $\Omega',\mblock{\Phi}, \mblock{\Phi,\Theta},\rig{\Diamond A}$. 
    In this case, the premise has a block with output formula but its input part is included in some other block, so it is not tight. 
    Assume w.l.o.g there are no other blocks such that $\Phi,\Theta$ is a subset of it. 
    By Proposition \ref{prop:tight-preserving}, we have $\Omega',\mblock{\Phi}, \mblock{\Phi,\Theta, \rig A}$, 
    By contraction, we have $\Omega',\mblock{\Phi,\Theta, \rig A}$. 
    We replace the application by another one with premise $\Omega',\mblock{\Phi,\Theta, \rig A}$ 
    and conclusion $\Omega', \mblock{\Phi,\Theta},\rig{\Diamond A}$. 
    Similar to the previous case, we reconstruct $\+D$ as follows: 
    For rule applications above the original premise $\Omega',\mblock{\Phi,\rig A}, \mblock{\Phi,\Theta}$, if it happens only in $\Omega'$, we keep as it is; 
    if it involves any of the two blocks $\mblock{\Phi,\rig A}, \mblock{\Phi,\Theta}$, we apply them only in the single one $\mblock{\Phi,\Theta, \rig A}$. 
    For rule applications below the conclusion $\Omega',\mblock{\Phi}, \mblock{\Phi,\Theta},\rig{\Diamond A}$, if it involves the block $\mblock{\Phi}$ we remove it as well as all the occurrences of $\mblock{\Phi}$ and keep the rest. 

    \noindent (3) $r=\lef {\square}$. Consider the following three cases: 
    $$
    \begin{aligned}
        \AXC{$\Omega',\lef {\square A}, \mblock{\Phi,\lef A}, \mblock{\Phi, \lef A,\rig F}$}
        \RightLabel{$\lef \square$}
        \UIC{$\Omega',\lef {\square A}, \mblock{\Phi}, \mblock{\Phi, \lef A,\rig F}$}
        \DP
    & \  &
        \AXC{$\Omega',\lef {\square A}, \mblock{\Phi,\lef A}, \mblock{\Phi, \lef A}$}
        \RightLabel{$\lef \square$}
        \UIC{$\Omega',\lef {\square A}, \mblock{\Phi}, \mblock{\Phi, \lef A}$}
        \DP
    & \ &
        \AXC{$\Omega',\lef {\square A}, \mblock{\Phi,\lef A}, \mblock{\lef A, \Phi, \rig F}$}
        \RightLabel{$\lef \square$}
        \UIC{$\Omega',\lef {\square A}, \mblock{\lef A,\Phi}, \mblock{\Phi, \rig F}$}
        \DP
    \end{aligned}
    $$
    For the first two cases we have the reconstruction as in (2) and (1) respectively. 
    Lastly for (3), 
    first by weakening and contraction we have $\Omega',\lef{\square A},\mblock{\Phi,\lef A,\rig F}$. 
    We replace the application above with 
    another one with premise $\Omega',\lef {\square A}, \mblock{\lef A,\Phi,\rig F}$ 
    and conclusion $\Omega',\lef {\square A}, \mblock{\Phi,\rig F}$ 
    and reconstruct the derivation as: 
    For rule applications above the original premise $\Omega',\lef {\square A}, \mblock{\Phi,\lef A}, \mblock{\lef A, \rig F}$, if it happens only in $\Omega'$, we keep as it is; 
    if it involves any of the two blocks $\mblock{\Phi,\lef A}, \mblock{\lef A, \Phi, \rig F}$, we apply them only in the single one $\mblock{\Phi,\lef A, \rig F}$. 
    For rule applications below the conclusion $\Omega',\lef {\square A}, \mblock{\lef A, \Phi}, \mblock{\Phi, \rig F}$, we remove these applications as well as all the occurrences of $\mblock{\lef A,\Phi}$ if they involves the block $\mblock{\lef A, \Phi}$ and keep the rest. 

    Each time we encounter the situations in (1)-(3), we reconstruct the proof until there are no any `bad' rule applications as above. In the end, we obtain a clean proof of some sub-sequent of $\Sigma$. 
\end{proof}

Now we consider proof search in order to find a \emph{clean proof} for any provable formula. To this end, we define the notion of redundancy in backward proof search. 

\begin{definition}[Redundant rule application]
    Let $(r)$ be a $\bullet$-rule in \cumlcfik. 
    An application $(r)$ to a sequent $\Omega$ on a formula $F$ (and, whenever is involved, on a block $\mblock{\Phi}\in\Omega$) is called \emph{redundant} if it satisfies one of the following conditions in Table \ref{fig:redundant-condition}. 
    \begin{table}[!t]
        \centering
        \renewcommand{\arraystretch}{1.1}
        \setlength{\tabcolsep}{5pt}
        \begin{tabular}{c|c|l}
            $r$ & $F$ & condition of redundancy\\
            \hline
            $\lef \vee_1$ & $\lef{A\vee B}$ & $\lef {A\vee B}\in\fm{\Omega}$, and $\lef A\in \fm{\Omega}$ or $\lef B\in \fm{\Omega}$\\
            $\lef \vee_2$ & $\lef{A\vee B}$ & $\lef{A\vee B}\in\Phi$, and $\lef A\in \Phi$ or $\lef B\in \Phi$\\
            $\lef {\supset_1}$ & $\lef {A\supset B}$ & $\lef {A\supset B}\in\fm{\Omega}$ and $\lef B\in\fm{\Omega}$\\
            $\lef {\supset_2}$ & $ \lef {A\supset B}$ & $\lef {A\supset B}\in\Phi$ and $\lef B\in \Phi$\\
            $\lef {\Diamond_1}$ & $\lef {\Diamond A}$ & $\lef {\Diamond A}\in\fm{\Omega}$ and $\lef A\in\Phi$ for some $[\Phi]\in \Omega$
        \end{tabular}
        \caption{\label{fig:redundant-condition} Conditions of redundancy for $(r)$ on $F$}
    \end{table}
\end{definition}

We adopt a (backward) proof search strategy with the following constraints: 
\begin{itemize}
    \item[(i)] no rule is applied to an axiom; \quad (ii) no rule is applied redundantly;
    \item[(iii)] no identical sequents occurring are allowed in the same branch;
    \item[(iv)] at each step, if the processed sequent is of the form $\Omega,\mblock{\Phi_1}, \mblock{\Phi_2}$ and $\Phi_1\subset \Phi_2$ then rule applications to $\Phi_1$ is forbidden. As a result, for a sequent of the form $\Omega, \mblock{\Phi_1},\ldots,\mblock{\Phi_k}$ with $\Phi_1\subset\ldots\subset\Phi_k$ then any rule acting on blocks \footnote{Namely all the $\bullet_2$-rules (except $\lef {\Diamond_2}$), together with $\rig\Diamond$ and $\lef{\square}$. } is only allowed to be applied to the maximal $\Phi_k$ or to the block containing the output formula. 
\end{itemize}
If the application of a rule produce at least one non-tight premise, the rule cannot be applied. 

By adopting these restriction, we obtain all sequents that can be generated in backward proof search are tight. 
This is the argument: in backward proof search, duplicated blocks can be produced only  in  the following cases: 
(a) by an application of  $\lef{\Diamond_1}$ to $\Omega,\lef{\Diamond A}, \mblock{\lef A}$; 
(b) by an application of  $\lef {\supset}$ to $\Omega, \mblock{\Phi},\mblock{\Phi,\rig F}$ which  removes $\rig F$ from the block;
(c) by an application of  $\bullet$-rules to $\mblock{\Phi}$ in $\Omega,\mblock{\Phi},\mblock{\Phi,\Theta}$ which expands $\Phi$ to $\Phi,\Theta$. 
Note that (a) is prevented by non-redundancy condition of $\lef {\Diamond_1}$ of  constraint (ii)  and (b) cannot occur by definition of tight sequent itself; and (c) is prevented indeed by the constraint (iv) above. 
The completeness of this constrained proof search is ensured by Proposition \ref{prop:clean-proof}.  

Now we turn to the final step, we determine the space requirement of proof search. Preliminarily , we have the following proposition whose proof is by an easy check of each rule. 

\begin{proposition}\label{prop:md-decrease}
    Let $(r)$ be a rule in \cumlcfik which is of the form $\frac{S_1\quad S_2}{S}$ or $\frac{S_1}{S}$. Then we have 
    $$
    \begin{aligned}
        \md{S_1}<\md{S} & \ & \text{if}~r\in\{\lef {\Diamond_2}, \rig{\supset_2},\rig{\square_2}\}; & \quad &
        \md{S_i}\leq \md{S} \ \text{for~}i\in \{1,2\} & \ & \text{otherwise.}
    \end{aligned}
    $$
\end{proposition}

\begin{restatable}{lemm}{termination}\label{lem:terminating}
    Backward proof search in \cumlcfik for a formula $A$ terminates with a finite clean derivation/proof where each branch is of exponential length in terms of $|A|$. 
\end{restatable}

\begin{proof}
    Let $A$ be a formula of size $n$, $\+D$ be a clean derivation of the tight sequent $\rig A$ and $\+B$ be an arbitrary branch in $\+D$. 
    We divide $\+B$ into phases $\+B_1,\ldots,\+B_k,\ldots$ where each $\+B_i=S_i^1,\ldots S_i^{m_i}$ such that $S_1^1= \rig A$ and each $S_{i+1}^1$ is obtained from $S_i^{m_i}$ by one of the rules $(\lef {\Diamond_2}), (\rig{\supset_2})$ and $(\rig{\square_2})$. This means the maximal index $k$ equals to the maximal number of applications of $(\lef {\Diamond_2}), (\rig{\supset_2})$ and $(\rig{\square_2})$ that are applied within $\+B$. 
    By Proposition \ref{prop:md-decrease}, $k\leq\md{A}$, hence 
    we can write $\+B=\+B_1,\ldots,\+B_k$ for some $k\leq\md{A}$. 
    Moreover, we claim 
        for any $i\in\{1,\ldots,k\}$, $m_i$ as the length of $\+B_i$ is bounded by $2^{O(n)}$. 
    
    \noindent \underline{Proof of claim:}
        $m_i$ as the length of $\+B_i$ is identical to the number of sequents that are produced by backward proof search. 
        Since $\+B_i$ does not contain 
        applications of 
        $(\lef {\Diamond_2}), (\rig{\supset_2})$ and $(\rig{\square_2})$, by definition, any backward rule application of other rules in \cumlcfik does not decrease the size of the input part. 
        Since each tight sequent is also a set-sequent, by Proposition \ref{prop:max-block}, there are at most $\+O(2^n)$ blocks contained in a tight sequent. Also, each block has size $\+O(n^2)$, as the input part is increasing within $\+B_i$, hence the number of different input parts for sequents in $\+B_i$ is $\+O(n^2\cdot 2^n)$. On the other hand, while keeping/increasing the set of input, the output formula might be replaced by $\lef \supset$ or decomposed by other $\circ$-rules. 
        The unique output formula can occur in any block and since there are $\+O(n)$ output formulas and $2^{\+O(n)}$ blocks, there are $2^{\+O(n)}\cdot \+O(n)$ possible occurrences of the output formula which give $\+O(2^n\cdot n^2)\cdot \+O(n\cdot 2^n)=2^{\+O(n)}$ different sequents, 
        which is the upper bound of $m_i$. 
        \hfill $\dashv$

    To estimate the length of each branch in $\+D$, since $\+B=\+B_1,\ldots,\+B_k$ and $k\leq \md{A}=\+O(n)$, by the claim above, 
    the length of $\+B$ is bounded by $2^{\+O(n)} \cdot \+O(n) = 2^{\+O(n)}$. Therefore, we conclude $\+B$ is of exponential length in terms of $n$. 
\end{proof}

By Lemma \ref{lem:terminating} and Proposition \ref{prop:max-block}, each branch in a derivation of $A$ has an exponential size of $|A|$. 
\begin{theorem}
    The decision problem of $\mfik$ is in \textsc{Expspace}. 
\end{theorem}

\section{Conclusion}

We have presented a `shallow' sequent calculus for intuitionistic modal logic \fik, which is a close variant or weakening of Fischer Servi/Simpson's \ik, regarded a  natural and meaningful \iml. 
By means of the shallow calculus, we have shown that the decision problem of \fik is in \textsc{Expspace}, much lower than the upper bound conjectured for \ik. 
However, whether this bound is tight, that is to say what is the lower bound for the decision problem of \fik, is an open problem at present. 

Form a proof-theoretic viewpoint, our `shallow' sequent calculi can be seen either as a minimal extension of standard Gentzen calculi, or as a drastic simplification of nested calculi. We believe that calculi of this type might be useful for studying certain meta-logical properties such as different kinds of interpolation. 
Moreover, we aim to develop similar calculi for other logics of the \iml family in order to get new or better complexity bounds. 

\nocite{*}
\bibliographystyle{eptcs}
\bibliography{refs}

\end{document}